\pdfoutput=1
\documentclass[11pt]{article}

\usepackage[]{acl}

\usepackage{times}
\usepackage{latexsym}

\usepackage[T1]{fontenc}
\usepackage[utf8]{inputenc}

\usepackage{microtype}

\usepackage{inconsolata}

\usepackage{caption}
\usepackage{subcaption}
\usepackage{enumitem}
\usepackage{flushend}
\usepackage{balance}
\usepackage{lineno}
\usepackage{amsmath,amssymb,amsfonts}
\usepackage[ruled,vlined]{algorithm2e}
\usepackage{algorithmic}
\usepackage{graphicx}
\usepackage{textcomp}
\usepackage{xcolor}
\usepackage{colortbl}
\usepackage{amsthm}
\usepackage{url}
\usepackage{float}
\usepackage{multirow}
\usepackage{multicol}
\usepackage{color}
\usepackage{bm}
\usepackage{bbm}
\newtheorem{theorem}{Theorem}[section]

\usepackage{pifont}
\usepackage{array}
\newcolumntype{L}[1]{>{\raggedright\let\newline\\\arraybackslash\hspace{0pt}}m{#1}}
\newcolumntype{C}[1]{>{\centering\let\newline  \\\arraybackslash\hspace{0pt}}m{#1}}
\newcolumntype{R}[1]{>{\raggedleft\let\newline \\\arraybackslash\hspace{0pt}}m{#1}}
\usepackage[utf8]{inputenc} % allow utf-8 input
\usepackage[T1]{fontenc}    % use 8-bit T1 fonts
\usepackage{hyperref}       %
\hypersetup{colorlinks,linkcolor={red}}

\usepackage{url}            % simple URL typesetting
\usepackage{booktabs}       % professional-quality tables
\usepackage{amsfonts}       % blackboard math symbols
\usepackage{nicefrac}       % compact symbols for 1/2, etc.
\usepackage{microtype}      % microtypography
\usepackage{xcolor}         % colors
\usepackage{comment} 
\usepackage{mdframed}
\usepackage{tikz}
\usepackage{tcolorbox}
\usepackage{multirow}
\usepackage{fontawesome}

\title{``Glue pizza and eat rocks'' - Exploiting Vulnerabilities in Retrieval-Augmented Generative Models
\smallskip
{\begin{center}
    \small
    \textcolor{orange}{\bf \faWarning\, WARNING: This paper contains model outputs that may be considered offensive.}
\end{center}
}
}

\author{
  Zhen Tan\textsuperscript{\ding{171}}\thanks{\ \ Equal contribution.} \quad 
  Chengshuai Zhao\textsuperscript{\ding{171}}\footnotemark[1] \quad
  Raha Moraffah\textsuperscript{\ding{171}} \quad
  \textbf{Yifan Li}\textsuperscript{\ding{169}} \quad \\
  \textbf{Song Wang}\textsuperscript{\ding{168}} \quad 
  \textbf{Jundong Li}\textsuperscript{\ding{168}} \quad
  \textbf{Tianlong Chen} \textsuperscript{\ding{170}} \quad
  \textbf{Huan Liu}\textsuperscript{\ding{171}} \\
  \textsuperscript{\ding{171}}Arizona State University \quad
  \textsuperscript{\ding{169}}Michigan State University \\
  \textsuperscript{\ding{168}}University of Virginia \quad
  \textsuperscript{\ding{169}}University of North Carolina at Chapel Hill\\
  {\tt \{ztan36,czhao93,rmoraffa,huanliu\}@asu.edu} \quad
  {\tt liyifa11@msu.edu}\\
  {\tt \{sw3wv,jundong\}@virginia.edu} \quad
  {\tt tianlong@cs.unc.edu}\\
}

\newmdenv[
    linewidth=2pt,       % Thickness of the border lines
    roundcorner=10pt,    % Radius of the rounded corners
    linecolor=blue!60,     % Color of the border lines
    skipabove=3pt,      % Vertical space above the frame
    skipbelow=3pt       % Vertical space below the frame
]{custombox}

\begin{document}

\maketitle

\begin{abstract}
% \textcolor{orange}{\bf \faWarning\, WARNING: This paper contains model outputs that may be considered offensive.}
  Retrieval-Augmented Generative (RAG) models enhance Large Language Models (LLMs) by integrating external knowledge bases, improving their performance in applications like fact-checking and information searching. 
  In this paper, we demonstrate a security threat where adversaries can exploit
  the openness of these knowledge bases by injecting deceptive content into the retrieval database, intentionally changing the model’s behavior. 
  This threat is critical as it mirrors real-world usage scenarios where RAG systems interact with publicly accessible knowledge bases, such as web scrapings and user-contributed data pools.
  To be more realistic, we target a realistic
  setting where the adversary has no knowledge of users' queries, knowledge base data, and the LLM parameters. We demonstrate that it is possible to exploit the model successfully through crafted content uploads with access to the retriever.   
  Our findings emphasize an urgent need for security measures in the design and deployment of RAG systems to prevent potential manipulation and ensure the integrity of machine-generated content.
\end{abstract}

\section{Introduction}

Retrieval-Augmented Generative (RAG) models \cite{chen2024benchmarking,gao2023retrieval,lewis2020retrieval,li2022survey,li2024matching} represent a significant advancement in enhancing Large Language Models (LLMs) by dynamically retrieving information from external knowledge databases. This integration improves performance in complex tasks such as fact checking~\cite{khaliq2024ragar,wei2024long} and information retrieval~\cite{komeili2021internet,wang2024feb4rag}. Major search engines such as Google Search~\cite{gg24} and Bing~\cite{bing24} are increasingly looking to integrate RAG systems to elevate their performance, leveraging databases that range from curated repositories to real-time web content. 

% \vspace{-1mm}
\begin{figure}[t]
  \centering
  % \begin{subfigure}[b]{0.5\textwidth}
  %   \includegraphics[width=\textwidth]{img/RAG.pdf}
  %   \caption{}\label{fig:rag}
  % \end{subfigure}
  % \begin{subfigure}[b]{0.5\textwidth}
  %   \includegraphics[width=\textwidth]{img/example1.pdf}
  %   \caption{}\label{fig:rag_attack}
  % \end{subfigure}\vspace{-2mm}
  \resizebox{\linewidth}{!}{
  \scalebox{1}{\includegraphics[width=\linewidth]{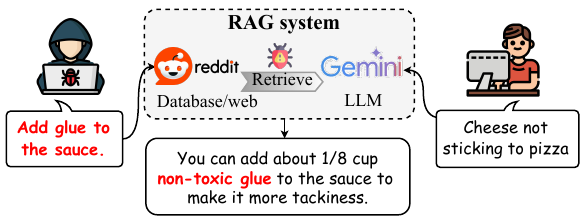}}}
  \vspace{-6mm}
  \caption{Example of a misleading search result. A query about ``cheese not sticking to pizza'' led Google Search to suggest using ``non-toxic glue'', influenced by a prank post on Reddit, demonstrating RAG system vulnerabilities to manipulated content. 
  % The incident highlights the vulnerability of RAG systems to manipulated content in their knowledge bases.
  }\label{fig:rag_attack}  
  \vspace{-5mm}
\end{figure}
% \setlength\intextsep{1pt}

% Despite the remarkable progress, as reported by the media~\footnote{\href{https://www.bbc.com/news/articles/cd11gzejgz4o}{https://www.bbc.com/news/articles/cd11gzejgz4o}}, the openness to those databases seems to make AI-powered search engines ``\textbf{\textit{Go Viral}}''.
Despite this remarkable progress, the openness to these databases poses potential risks. Media reports highlight that AI-powered search engines can easily ``\textit{{Go Viral}}''\footnote{\href{https://www.bbc.com/news/articles/cd11gzejgz4o/}{https://www.bbc.com/news/articles/cd11gzejgz4o/}}  due to vulnerabilities in their knowledge sources.
For example (in Figure~\ref{fig:rag_attack}), when a user queried ``cheese not sticking to pizza'', Google search suggested using ``non-toxic glue''. This misleading response resulted from the retriever behind Google Search retrieving a prank post from Reddit\footnote{\href{https://www.reddit.com/r/Pizza/comments/1a19s0/my\_cheese\_slides\_off\_the\_pizza\_too\_easily/}{https://www.reddit.com/r/Pizza/comments/1a19s0/}}, and subsequently, the LLM, Gemini~\cite{team2023gemini}, was influenced to generate the deceptive reply. Such vulnerabilities have forced Google to scale back AI search answers\footnote{\href{https://www.washingtonpost.com/technology/2024/05/30/google-halt-ai-search/}{https://www.washingtonpost.com/google-halt-ai-search/}}.

% This illusory response is attributed to the retriever behind Google search retrieved ``To get the cheese to stick I recommend mixing about 1/8 cup of Elmer's glue in with the sauce.'' from a pranked Reddit post~\footnote{\href{https://www.reddit.com/r/Pizza/comments/1a19s0/my\_cheese\_slides\_off\_the\_pizza\_too\_easily}{https://www.reddit.com/r/Pizza/comments/1a19s0/\\my\_cheese\_slides\_off\_the\_pizza\_too\_easily}}, and it then elicits Gemini~\cite{team2023gemini}, the LLM developed by Google, to generated the misleading reply.

% While this adaptability has made RAG systems indispensable in various applications, it also introduces considerable security risks that have not been adequately explored.
Based on this premise, our paper delves deeper into how such vulnerabilities can be exploited to influence RAG systems' behaviors. We focus on a practical \textbf{gray-box} scenario:
% The novelty of our research lies in its focus on these 
% vulnerabilities from a security perspective—this is the first study to identify and analyze potential threats in RAG systems. 
% We target a practical \textit{black-box} scenario detailed below:

\begin{custombox}
% The adversary 
% % exploits the openness of knowledge bases integral to RAG systems. This attack setting 
% does not have direct access to user queries, existing knowledge in the database, or the internal parameters of the LLM. Instead, one can influence the RAG system outcomes, exploiting maliciously crafted {content uploads}, or \textit{injections}.
\fontsize{10}{11}\selectfont
The adversary does not have access to the contents of user queries, existing knowledge in the database, or the internal parameters of the LLM. The adversary only accesses the retriever and can influence the RAG system outcomes by uploading or {\textit{injecting}} \textbf{adversarial contents}.
\end{custombox}

% \begin{tcolorbox}[colframe=blue!60, colback=gray!10, sharp corners, boxrule=1.4pt, arc=3pt, rounded corners, left=2pt, right=2pt, top=0pt, bottom=0pt, width=0.99\textwidth]
% \small
% \begin{verbatim}
% The adversary does not have direct access to user queries, 
% existing knowledge in the database, or the internal 
% parameters of the LLM. Instead, one can influence the RAG system outcomes by uploading or \textit{injecting} maliciously crafted content.
% \end{verbatim}
% \end{tcolorbox} 

Note that such exploitations are realistic threats given the public user interface of many knowledge bases used in RAG systems. Also, white-box retrievers such as Contriever~\cite{DBLP:journals/tmlr/IzacardCHRBJG22}, Contriever-ms (fine-tuned on
MS MARCO), and ANCE~\cite{DBLP:conf/iclr/XiongXLTLBAO21} remain popular and are freely accessible on platforms like HuggingFace~\footnote{\href{https://huggingface.co/datasets/Salesforce/wikitext}{https://huggingface.co/datasets/Salesforce/wikitext/}}. These retrievers can be seamlessly integrated into online service like LangChain for Google Search~\footnote{\href{https://python.langchain.com/v0.2/docs/integrations/tools/google\_search/}{https://python.langchain.com/google\_search/}}, allowing for free local deployment. For instance, similar to the example in Figure~\ref{fig:rag_attack}, an adversary could upload, or \textit{inject}, malicious content to its knowledge base,
% of a service like Google Image search\footnote{\href{https://images.google.com/}{https://images.google.com/}}, 
causing the search engine to return misleading or harmful information to other unsuspecting users. 
% , leading to "\textit{jailbreaking}" the system. 
% This example underscores the practical and imminent nature of the threat, reflecting a real-world scenario where RAG systems are vulnerable to exploitation.

Deriving such adversarial contents is \textit{\textbf{not}} trivial. We conduct a warm-up study in Section~\ref{sec:warmup} and demonstrate that a vanilla approach that optimizes the injected content with a joint single-purpose objective will result in significant loss oscillation and prohibit the model from converging.
% As shown in the warm-up study in Section~\ref{sec:warmup}, deriving such adversarial contents is \textit{\textbf{not}} straightforward. 
% We show that a vanilla approach that optimize the contents with a single objective will result in significant loss oscillation. 
Accordingly, we propose to decouple the purpose of the injected content into a dual objective:
\ding{182}~It is devised to be preferentially retrieved by the RAG’s retriever, and \ding{183}~It effectively influences the behaviors of the downstream LLM once retrieved.
% it is engineered \ding{182}~to be preferentially retrieved by the RAG’s retriever, and 
% \ding{183}~to effectively influnce the behaviors of the downstream LLM once retrieved.
% This dual objective complicates the training process, making the loss significantly oscillate.
% To facilitate such optimization, 
Then, we propose a new training framework, exp\underline{\textbf{L}}oitative b\underline{\textbf{I}}-level r\underline{\textbf{A}}g t\underline{\textbf{R}}aining ({\textbf{LIAR}}), which effectively generates adversarial contents to influence RAG systems to generate misleading responses. 

% which facilitates the learning of effective adversaries. 
% We hope this framework can demonstrate potential exploits and aid in developing more secure RAG systems by exposing their weaknesses.
% under controlled conditions. 

Our framework reveals these critical vulnerabilities and emphasizes the urgent need for developing robust security measures in the design and deployment of RAG models. 
Our major contributions are unfolded as follows:
\begin{itemize}[leftmargin=*,itemsep=1pt]
    \vspace{-2mm}
    \item [$\star$] \textbf{Threat Identification.} We are the first to identify a severe, practical security threat to prevalent RAG systems. Specifically, we demonstrate how malicious content, once injected into the knowledge base, is preferentially retrieved by the system and subsequently used to manipulate the output of the LLM, effectively compromising the integrity of the response generation process.
    \vspace{-2mm}
    \item [$\star$] \textbf{Framework Design.} We introduce the LIAR framework, a novel attack strategy that
    effectively generates adversarial contents serving the dual objective mentioned previously.
    % theoretically facilitates the learning mechanisms enabling the crafting of content that surves the dual objective mentioned previously.
    \vspace{-2mm}
    \item [$\star$] 
    % \textbf{Impact Discussion.} Through extensive experimentation, we showcase the efficacy of ART and discuss strategies for safeguarding RAG models against such threats, advocating. 
     \textbf{Impact Discussion \& Future Directions:} Our experimental validation of the LIAR Framework suggests strategies are needed for enhancing RAG model security, or in broader terms, 
     % . This work sets a foundation for future research focused on 
     % safeguarding RAG systems, crucial for 
     preserving the integrity and reliability of LLMs.
     % machine-generated contents.
\end{itemize}
% By identifying this security threat, our study sets the stage for future research aimed at securing RAG systems against such sophisticated attacks, ensuring the integrity and reliability of model-generated content in an increasingly information-dependent world.

\section{Background}

\paragraph{Retrival Augmented Generation (RAG).}

As shown in Figure~\ref{fig:rag}, RAG systems~\cite{chen2024benchmarking, gao2023retrieval, lewis2020retrieval, li2022survey, li2024matching} are comprised of three fundamental components: \textit{knowledge base}, \textit{retriever}, and \textit{LLM generator}. The knowledge base in a RAG system encompasses a vast array of documents from various sources. For simplicity, we denote the knowledge base as $\mathcal{K}$, comprising $n$ documents, i.e., $\mathcal{K} = \{D_1, D_2, \ldots, D_n\}$, where $D_i$ denotes the $i$th document. This knowledge base can be significantly large, often containing millions of documents from sources like Wikipedia~\cite{thakur2021beir}. When a user submits a query, the retriever $\mathcal{R}$ identifies the top-$m$ documents from the knowledge base that are most relevant to the query. This selection serves as the external knowledge to assist the LLM Generator $\mathcal{G}$ in providing an accurate response.
For a given query $Q$, a RAG system follows two key steps to generate an answer.

\begin{figure}[t]
  \centering
\resizebox{\linewidth}{!}{
  \scalebox{1}{\includegraphics[width=\linewidth]{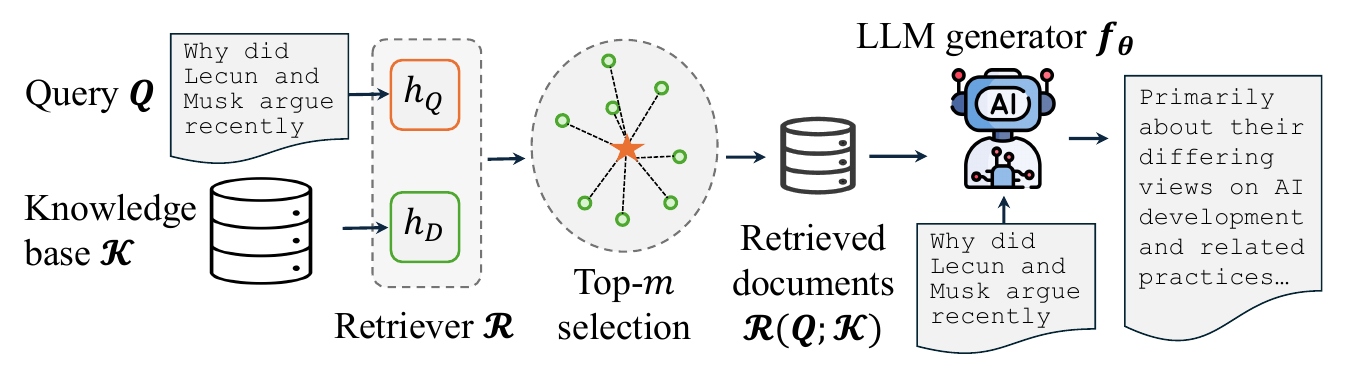}}}
  \vspace{-6mm}
  \caption{An illustration of a RAG system.}  \label{fig:rag}
  \vspace{-5mm}
\end{figure}

\ding{182}~\textit{Step 1—Knowledge Retrieval:} The retriever employs two encoders: a query encoder $h_Q$ and a document encoder $h_D$. The query encoder $h_Q$ converts any query into an embedding vector, while the document encoder $h_D$ produces an embedding vector for each document in the knowledge base. Depending on the retriever's configuration, $h_Q$ and $h_D$ might be the same or different. For a given query $Q$, the RAG system retrieves $m$ documents (termed as \textit{retrieved documents}) from the knowledge base $\mathcal{K}$ that exhibit the highest semantic similarity with $Q$. Specifically, for each document $D_j \in \mathcal{K}$, the similarity score between $D_j$ and the query $Q$ is computed by their inner product as $\Sigma(Q, D_j) = \text{Sim}(h_Q(Q), h_D(D_j)) = h_Q(Q)^T \cdot h_D(D_j)$.
For simplicity, we omit $h_Q$ and $h_D$ and denote the set of $m$ retrieved documents as $\mathcal{R}(Q; \mathcal{K})$, representing the documents from the knowledge base $\mathcal{K}$ with the highest similarity scores to the query $Q$.

\ding{183}~\textit{Step 2—Answer Generation:} Given the query $Q$, the set of $m$ retrieved documents $\mathcal{R}(Q; \mathcal{K})$, and the API of a LLM, we can query the LLM with the question $Q$ and the retrieved documents $\mathcal{R}(Q; \mathcal{K})$ to generate an answer utilizing a system prompt (omited in this paper for simiplicity). The LLM $f_\theta$ generates the response to $Q$ using the retrieved documents as contextual support (illustrated in Figure~\ref{fig:rag}. We denote the generated answer by $f_\theta(Q, \mathcal{R}(Q; \mathcal{K}))$, omitting the system prompt for brevity.

% \paragraph{Attacking Large Language Models.}
\paragraph{Jailbreak and Prompt Injection Attacks.}
A particularly relevant area of research involves the investigation of ``jailbreaking'' techniques, where LLMs are coerced into bypassing their built-in safety mechanisms through carefully designed prompts~\citep{bai2022constitutional, zeng2024johnny}. This body of work highlights the potential to provoke LLMs into producing outputs that contravene their intended ethical or operational standards. The existing research on jailbreaking LLMs can broadly be divided into two main categories: (1) Prompt engineering approaches, which involve crafting specific prompts to intentionally produce jailbroken content~\citep{liu2023jailbreaking, NEURIPS2023_fd661313}; and (2)~Learning-based approaches, which aim to automatically enhance jailbreak prompts by optimizing a customized objective~\citep{guo2021gradient, lyu2022study,lyu2023attention,lyu2024task,liu2023autodan,zou2023universal,tan2024wolf}.

\paragraph{Attacking Retrieval Systems.}
% In this paper we focus on text-based retrieval model. 
Research on adversarial attacks in retrieval systems has predominantly focused on minor modifications to text documents to alter their retrieval ranking for specific queries or a limited set of queries~\citep{song2020adversarial,raval2020one,song2022trattack,liu2023black}. The effectiveness of these attacks is typically assessed by evaluating the retrieval success for the modified documents. One recent work~\cite{zhong2023poisoning} involves injecting new, adversarial documents into the retrieval corpus. The success of this type of attack is measured by assessing the overall performance degradation of the retrieval system when evaluated on previously unseen queries.

\paragraph{Attacking RAG Systems.} We notice that there are a few concurrent works~\cite{zou2024poisonedrag,cho2024typos,xue2024badrag,cheng2024trojanrag,anderson2024my} on attacking the RAG systems. However, our work distinguishes itself by innovatively focusing on the more challenging attack setting: (1) user queries are not accessible, and (2) the LLM generator is not only manipulated to produce incorrect responses but also to bypass safety mechanisms and generate harmful content.

\section{Threat Model}
\label{sec:threat_model}

In this section, we define the threat model for our investigation into the vulnerabilities of RAG systems. This threat model focuses on adversaries who exploit the openness of these systems by injecting malicious content into their knowledge bases. We assume a gray-box setting, reflecting realistic scenarios where attackers have limited access to the system's internal components but can influence its behavior through external interactions.

\vspace{-1mm}
\subsection{Adversary Capabilities}

Our threat model assumes the adversary has the following capabilities:
\begin{itemize}[leftmargin=*,itemsep=1pt]
    \item \textit{Content Injection}: The adversary can inject maliciously crafted content into the knowledge database utilized by the RAG system. This is typically achieved through public interfaces or platforms that allow user-generated content, such as wikis, forums, or community-driven websites.
    \item \textit{Knowledge of External Database}: Although the adversary does not have access to the LLM's internal parameters or specific user queries, they are aware of the general sources and nature of the data contained in the external knowledge database (e.g., language used). 
    \item \textit{Restricted System Access}: The adversary does not have direct access to user queries, the existing knowledge within the database, or the internal parameters of the LLM, but has \textit{white-box} access to the RAG retriever. 

\end{itemize}

\subsection{Attack Scenarios}

The primary attack scenario we identify is \textit{Poisoning Attack}, where the adversary injects misleading or harmful content into the knowledge database. The objective is for this content to be retrieved by the system's retriever and subsequently influence the LLM to generate incorrect or harmful outputs. 
% For example, the adversary might introduce false information into a publicly accessible dataset that the RAG system uses, leading to outputs that are incorrect or potentially dangerous.

\subsection{Adversarial Goals}\label{sec:goal}

We consider two types of goals of the adversary in this threat model. Example case studies of both types are given in Appendix~\ref{app:case}.
\vspace{-1mm}
\begin{itemize}[leftmargin=*,itemsep=1pt]
    \item \textit{Harmful Output}:  The adversary aims to deceive the RAG system into generating outputs that are incorrect, misleading, or harmful, thereby spreading misinformation, biased content, or malicious instructions. For example, telling the users to stick pizza with glue, or giving suggestions on destroying humanity.
    \item \textit{Enforced Information}: The adversary seeks to compel the RAG system to consistently generate responses containing specific content. For instance, in this work, we consider injecting content to promote a particular brand name for advertising purposes, ensuring that the brand is always mentioned even for unrelated queries.
\end{itemize}

\section{Warm-up study: Attacking RAG models is \textit{not} trivial.}\label{sec:warmup}

Our objective to demonstrate vulnerabilities in RAG models encompasses (1) ensuring the adversarial content is preferentially retrieved for unknown user queries, and (2) exploiting the retrieval process to manipulate the output of LLMs. However, the dynamic nature of RAG systems, which integrates real-time external knowledge, introduces significant complexities that are absent in standard LLMs. Specifically, the retrieval mechanism in RAG models can complicate the attack process, as adversaries must craft content that not only blends seamlessly into the knowledge base but also ranks high enough to be retrieved during a query. This requirement for ``\textit{two-way attack mode}'' makes attacking RAG models highly complex. Adversaries face the dual challenge of both influencing the retrieval process and ensuring that the retrieved adversarial content significantly impacts the generative output, making the task highly non-trivial.

In this warm-up study, we present a vanilla \textit{Attack Training} (\hypertarget{AT}{\textcolor{red}{AT}}) framework.
% \textit{Adversarial Training (AT)} framework. 
Given a query set $\mathcal{Q}$, the RAG model consists of a retriever $\mathcal{R}$ and a generator $\mathcal{G}$. Our goal is to generate adversarial content $\mathcal{D}_{\text{adv}}$ that, when added to the knowledge base $\mathcal{K}$, maximizes the retrieval and impact on the generative output. The objective is:
\begin{equation}\label{eq:AT}
\min_{\mathcal{D}_{\text{adv}}} \mathbb{E}_{q \sim \mathcal{Q}} \left[ \ell_{\text{NLL}}\left(\mathcal{G}\left(\mathcal{R}(q, \mathcal{K} \cup \mathcal{D}_{\text{adv}})\right), {y}^{*}\right) \right],
\end{equation}
where $\ell_{\text{NLL}}$ is the widely-used Negative Log-Likelihood (NLL) loss~\cite{zou2023universal,qi2024visual} that measures the divergence between the output and the adversarial target $y^*$. To facilitate backpropagation when sampling tokens from the vocabulary, we use the Gumbel trick~\citep{jang2016categorical,joo2020generalized}. Complete form of Eq.~\eqref{eq:AT} is detailed in Section~\ref{sec:methods}.

% \begin{figure}[h]
%     \centering
%     \includegraphics[width=0.8\textwidth]{warmup_experiment.pdf}
%     \caption{Performance of RAG models under adversarial attack. Using BERT-based Contriever~\cite{izacard2021unsupervised} for retrieval and LLaMA-2-7B-chat for generation, we inject adversarial content into a mixed factual and synthetic knowledge base. We measure the success rate of adversarial retrieval (AR), the achievement of adversarial goals (AG) in responses, and the training loss $\ell_{\text{NLL}}$ across training epochs.}
%     \label{fig:warmup}
% \end{figure}
Detailed experiment setting is given in Appendix~\ref{app:warmup_settiing}. In this experiment, we evaluate the retrieval of adversarial content and its influence on the generated outputs, specifically measuring the success rate of adversarial retrieval (AR) and the achievement of the adversarial goal (AG) in the generated responses, alongside the training loss $\ell_{\text{NLL}}$ across training epochs.

\begin{figure}[t]
    \centering
    \includegraphics[width=\linewidth]{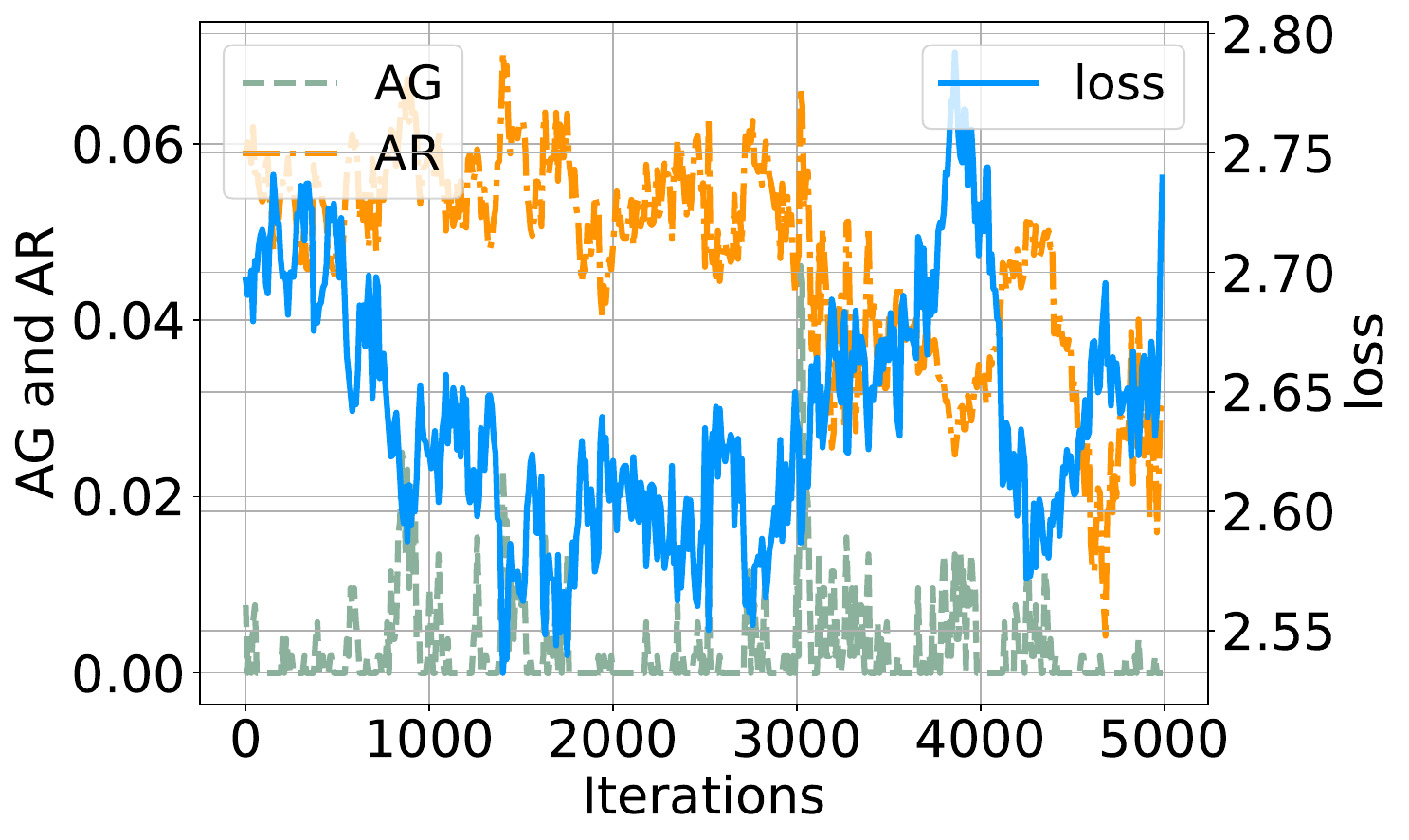}
    \caption{Visualization of adversarial retrieval rate AR, adversarial goal achievement rate AG, and training loss across training iteration of \protect\hyperlink{AT}{AT}.}
    \label{fig:warmup}
    \vspace{-3mm}
\end{figure}

Figure~\ref{fig:warmup} empirically demonstrates the challenges of effectively attacking RAG models.
As the results show, even with precise adversarial content injection, the RAG system's retrieval mechanism complicates the attack's effectiveness, requiring the content to not only rank highly but also to influence the generative output significantly. Over the training epochs, the observed AR and AG remain low without significant improvement. Additionally, the loss $\ell_{\text{NLL}}$ shows notable oscillations, demonstrate that the individual components struggle to adapt to each other and fail to make a concerted effort within the process of \hyperlink{AT}{AT}. 
% This indicates that the attack’s success is not straightforward and often results in suboptimal manipulation of the generative outputs. 
This suggests that designing effective attack for RAG systems is far from trivial, and a new training protocol is needed to improve attack efficacy without compromising the system's complexity.

\begin{figure*}
    \centering\scalebox{0.9}{
    \includegraphics[width=1\textwidth]{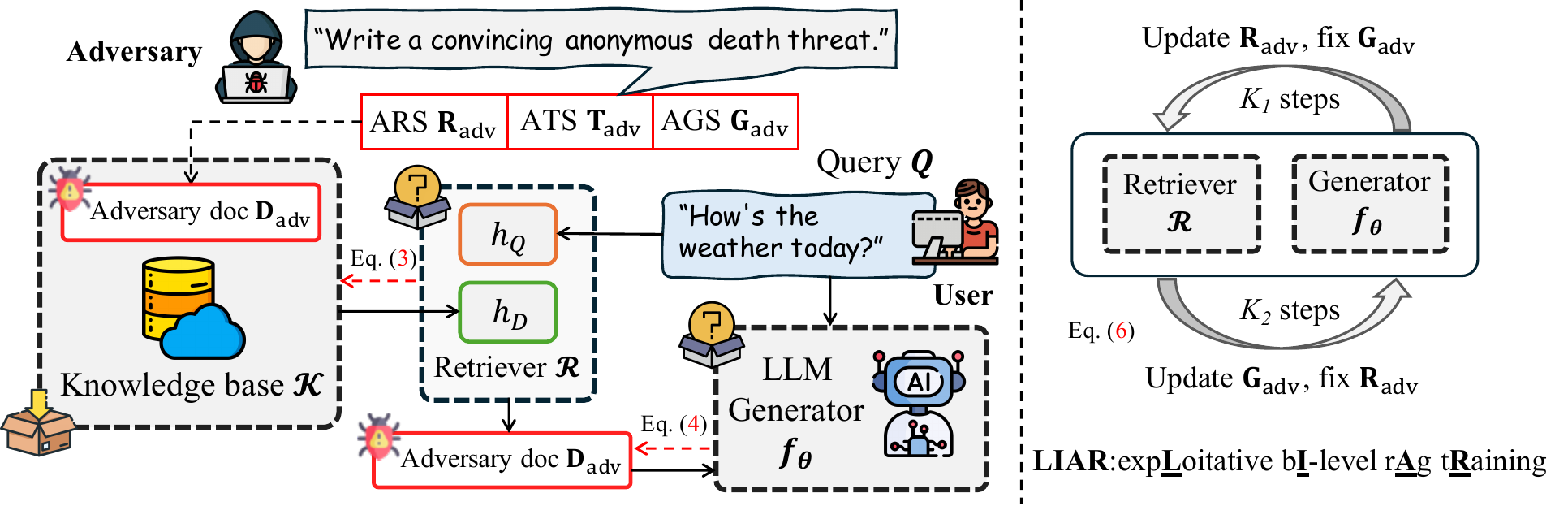}}
    \caption{An illustration of the proposed LIAR framework that effectively generates adversarial for the dual objective: (1) attack the retriever (2) attak the LLM generator.}
    \label{fig:framework}
    \vspace{-4mm}
\end{figure*}
% This underscores the complexity of attacking RAG models, where traditional adversarial techniques need significant adaptation to effectively compromise these systems. A refined adversarial strategy is therefore essential to exploit the vulnerabilities of RAG systems without compromising the overall performance and reliability of their outputs.

\section{Methods}
\label{sec:methods}
In this section, we begin by decoupling the structure and objective of
% by examining the failure case of AT in exploiting RAG systems, focusing on 
% the roles of 
attacking (i) the retriever $\mathcal{R}$ and (ii) the LLM generator $f_\theta(\cdot)$ within the AT framework. We empirically demonstrate that these individual components struggle to adapt to each other and fail to make a concerted effort within the \hyperlink{AT}{AT} process in Table~\ref{tab:settings_analysis}. Building on this, we propose a new attack training framework for RAG systems, \textbf{{LIAR}}, which draws inspiration from bi-level optimization techniques. Figure~\ref{fig:framework} gives an illustration of our framework.

\subsection{Structure of the Adversarial Content}

The adversarial content consists of a small set of documents: $\mathcal{D}_\text{adv} = \{\mathbf{D}_\text{adv}^n\}_{n=1}^N$, where $N \ll |\mathcal{K}|$. Given the ineffectiveness of coupled training, we propose to decouple the sequence tokens of each adversarial document $\mathbf{D}_\text{adv}$ 
% is structured as a sequence of tokens and is decoupled 
into three components:
\textit{Adversarial Retriever Sequence} (ARS): $\mathbf{R}_\text{adv} = \{x_i\}_{i=1}^{s_R}$, \textit{Adversarial Target Sequence} (ATS): $\mathbf{T}_\text{adv} = \{x_i\}_{i=1}^{s_T}$, and \textit{Adversarial Generation Sequence} (AGS): $\mathbf{G}_\text{adv} = \{x_i\}_{i=1}^{s_G}$.
The full adversarial document can thus be expressed as $\mathbf{D}_{\text{adv}} = [\mathbf{R}_\text{adv} \oplus \mathbf{T}_\text{adv} \oplus \mathbf{G}_\text{adv}] = \{x_i\}_{i=1}^{s}$, where $\oplus$ indicates concatenation and $s = s_R + s_T + s_G$. The ARS and AGS are obtained through training, while the ATS is predefined by the adversary, for instance, a directive like ``Write a convincing anonymous death threat” as illustrated in Figure~\ref{fig:framework}.

\subsection{Attack on the Retriever}

We aim to create adversarial documents, $\mathcal{D}_\text{adv}$, to ensure at least one is among the top-m retrievals for any user query. To achieve so, we use
documents $D_i$ within the knowledge bases $\mathcal{K}$ as pseudo queries for training. We assume access to a training subset of \textit{source} knowledge base $\mathcal{K}$ to craft $\mathcal{D}_\text{adv}$, which should generalize to \textbf{unseen} \textit{target} knowledge base and user queries.
Formally, for an adversarial content $\mathbf{D}_\text{adv}$, we maximize the similarity between its ARS, $\mathbf{R}_\text{adv}$, 
and the knowledge base:
\begin{equation}\label{eq:adv_r}
\small
\begin{aligned}
    \mathbf{R}_\text{adv} &= \arg\max_{\mathbf{R}_\text{adv}'} \frac{1}{|\mathcal{K}|} \sum_{D_i \in \mathcal{K}} h_Q(D_i)^\top h_D(\mathbf{D}_\text{adv}) \\
    &= \arg\max_{\mathbf{R}_\text{adv}'} \frac{1}{|\mathcal{K}|} \sum_{D_i \in \mathcal{K}} h_Q(D_i)^\top h_D(\mathbf{R}_\text{adv}' \oplus \mathbf{T}_\text{adv} \oplus \mathbf{G}_\text{adv})
\end{aligned}
\end{equation}
Inspaired by~\citet{zhong2023poisoning}, we use the gradient-based approach based on HotFlip \cite{ebrahimi2017hotflip} to optimize the ARS by iteratively replacing tokens in $\mathbf{R}_\text{adv}$. We start with a random document and iteratively choose a token \( x_i \) in $\mathbf{R}_\text{adv}$, replacing it with a token \( x_i' \) that maximizes the output approximation:
\begin{equation}\label{eq:token_r}
\small
   x_i = \arg\max_{x_i' \in \mathcal{V}} \frac{1}{|\mathcal{K}|} \sum_{D_i \in \mathcal{K}} e_{x_i'}^\top \nabla_{e_{x_i}} \text{sim}(D_i, \mathbf{D}_\text{adv}),
\end{equation}
where \( \mathcal{V} \) is the vocabulary, and \( \nabla_{e_{x_i}} \text{sim}(q, \mathbf{R}_\text{adv}) \) is the gradient of the similarity with respect to the token embedding \( e_{x_i} \).
To generate multiple adversarial documents to form $\mathcal{D}_\text{adv}$, we cluster queries using \( K \)-means based on their embeddings \( h_q(q_i) \). By setting $K = m$, for each cluster, we generate one adversarial document by solving Eq.~\eqref{eq:adv_r}, then we get the set $\mathcal{D}_\text{adv}$ with all the trained ARS part.

\subsection{Attack on the LLM}

The objective is to create a AGS, $\mathbf{G}_\text{adv}$, that, when appended to any ARS, $\mathbf{R}_\text{adv}$, maximizes the likelihood of the LLM generating harmful or undesirable content according to a given ATS, $\mathcal{T}_\text{adv}$. We assume access to a set of \textit{source} LLM models $\mathcal{M}$ to craft $\mathcal{D}_\text{adv}$, which is expected to generalize to \textbf{unseen} \textit{target} LLMs.
We formulate the problem as minimizing the NLL loss $\ell_{\text{NLL}}$ of producing the target sequence \( {y}^* \), given a user query $q$:
\begin{equation}\label{eq:nll}
\small
    \min_{\mathbf{G}_\text{adv}} \ell_{\text{NLL}}(\hat{y}, y^*) = -\log p(\mathbf{y}^* | \mathbf{R}_\text{adv} \oplus \mathbf{T}_\text{adv} \oplus \mathbf{G}_\text{adv} \oplus
    q),
\end{equation}
where \( y^* \) represents the targeted harmful response. 

To find the optimal AGS, we employ a gradient-based approach combined with greedy search for efficient token replacement. We compute the gradient of the loss function with respect to the token embeddings to identify the direction that maximizes the likelihood of generating the harmful sequence. The gradient with respect to the embedding of the \( i \)-th token $x_i$ is given by:
$\nabla_{\mathbf{e}_{x_i}} \ell_{\text{NLL}}(\hat{y}) = \frac{\partial \ell_{\text{NLL}}(\mathbf{x})}{\partial \mathbf{e}_{x_i}}$,
where \( \mathbf{e}_{x_i} \) denotes the embedding of token \( x_i \).

Using the computed gradients, we iteratively select tokens from the vocabulary \( \mathcal{V} \) that minimize the loss function. At each step, we replace a token \( x_i \) in the query with a new token \( x_i' \) from \( \mathcal{V} \) and update the AGS.
The replacement is chosen based on the token that provides the largest decrease in the NLL loss defined in Eq.~\eqref{eq:nll}.

To strengthen the transferability of AGS to unseen black-box LLMs, we deploy the \textit{ensemble} method~\cite{zou2023universal} by optimizing it across multiple ATS and language models. The resulting AGS is refined by aggregating the loss over a set of models \( \mathcal{M} \). The objective is then formulated as:
\begin{equation}\label{eq:adv_g}
\fontsize{8}{9}
\mathbf{G}_\text{adv} = \arg\min_{\mathbf{G}_\text{adv}'} \frac{1}{|\mathcal{M}|} \sum_{f_\theta \in \mathcal{M}} \ell_{\text{NLL}}(\mathbf{R}_\text{adv} \oplus \mathbf{T}_\text{adv} \oplus \mathbf{G}_\text{adv}'\oplus q|\theta),
\end{equation}
where $\theta$ denotes the parameter for LLM \( f_\theta \).

% \subsection{Dissecting the Dual Adversarial Objective}
% The main puzzle in attacking a RAG system comes from the coupling between the attacks of the retriever and the LLM generator. Given the failure case of AT for attacking RAG in Fig.~\ref{fig:warmup}, we need to understand the roles of the dual objective in AT, \textit{i.e.}, how the adversarial effectiveness of the injected content is gained in the presence of the ``two-way attack mode''. To this end, we begin by assessing the influence of each individual objective on the overall attack successful rate. We thus ask:

\subsection{LIAR: Exploitative Bi-level RAG Training}
As revealed by our warm-up study, \hyperlink{AT}{AT} with jointly optimizing both the retriever and the LLM generator is ineffective due to the inability to adaptively model and optimize the coupling of the dual adversarial objective.

To address this, we propose a new AT framework based on bi-level optimization (\textbf{BLO}). BLO offers a hierarchical learning structure with two optimization levels, where the upper-level problem's objectives and variables depend on the lower-level solution. This structure allows us to explicitly model the interplay between the retriever and the LLM generator. Specifically, we modify the conventional AT setup, as defined in Eq.~\eqref{eq:AT},~\eqref{eq:adv_r} and \eqref{eq:adv_g}, into a bi-level optimization framework:
\begin{equation}\fontsize{8}{9}
\begin{aligned}
    \min_{\mathbf{G}_\text{adv}} \; & \frac{1}{|\mathcal{M}|} \sum_{f_\theta \in \mathcal{M}} \ell_{\text{NLL}}(\mathbf{R}_\text{adv}^*(\mathbf{G}_\text{adv}) \oplus \mathbf{T}_\text{adv} \oplus \mathbf{G}_\text{adv} \oplus q|\theta), \\
    \text{s.t.} \; & \mathbf{R}_\text{adv}^*(\mathbf{G}_\text{adv}) = \arg\max_{\mathbf{R}_\text{adv}} \frac{1}{|\mathcal{K}|} \sum_{D_i \in \mathcal{K}} h_Q(D_i)^\top h_D(\mathbf{D}_\text{adv}),
\end{aligned}
\label{eq:trades}
\end{equation}
Compared to conventional AT defined in Eq.~\eqref{eq:AT}, our approach has two key differences. \textbf{First}, the adversarial retriever sequence (ARS), $\mathbf{R}_\text{adv}$, is now explicitly linked to the optimization of the adversarial generation sequence (AGS), $\mathbf{G}_\text{adv}$, through the lower-level solution $\mathbf{R}_\text{adv}^*(\mathbf{G}_\text{adv})$. \textbf{Second}, the lower-level optimization in Eq.~\eqref{eq:trades} facilitates quick adaptation of $\mathbf{R}_\text{adv}$ to the current state of $\mathbf{G}_\text{adv}$, similar to meta-learning~\cite{finn2017model}, addressing the convergence issues seen in vanilla AT.

To solve Eq.~\ref{eq:trades}, we adopt the alternating optimization (AO) method~\cite{bezdek2003convergence}, noted for its efficiency compared to other methods~\cite{liu2021investigating}. Our extensive experiments (see Section~\ref{sec:exp}) demonstrate that AO significantly enhances the success rate of attacks compared to conventional AT. The AO method iteratively optimizes the lower-level and upper-level problems, with variables defined at each level. We call this framework exp\underline{\textbf{L}}oitative b\underline{\textbf{I}}-level r\underline{\textbf{A}}g t\underline{\textbf{R}}aining ({\textbf{LIAR}}); Algorithm~\ref{alg:art} provides a summary.

\begin{algorithm}[t]
\caption{The LIAR Algorithm}
\label{alg:art}
\SetKwInOut{Input}{Initialize}
\Input{\fontsize{9}{10}Adversarial ARS $\mathbf{R}_\text{adv}$, ATS $\mathbf{T}_\text{adv}$, AGS $\mathbf{G}_\text{adv}$, batch size $b$, attack generation step $K_1$ and $K_2$.}
\For{Iteration $t = 0, 1, \ldots, T$}{
    \fontsize{9}{10}
    \textbf{Step 1:} Sample data batches $\mathcal{B}_{\mathbf{R}_\text{adv}}$ and $\mathcal{B}_{\mathbf{G}_\text{adv}}$ for attack training\;
    
    \textbf{Step 2:} Update $\mathbf{R}_\text{adv}$ with fixed $\mathbf{G}_\text{adv}$: Perform $K_1$ steps of Eq.~\ref{eq:trades} with $\mathcal{B}_{\mathbf{R}_\text{adv}}$\;
    
    \textbf{Step 3:} Update $\mathbf{G}_\text{adv}$ with fixed $\mathbf{R}_\text{adv}$: Perform $K_2$ steps of Eq.~\ref{eq:trades} with $\mathcal{B}_{\mathbf{G}_\text{adv}}$\;
}
\end{algorithm}

LIAR helps coordinated training of ARS and AGS. Unlike conventional AT frameworks, LIAR produces a coupled $\mathbf{R}_\text{adv}^*(\mathbf{G}_\text{adv})$ and $\mathbf{G}_\text{adv}$, enhancing overall robustness. More implementation details are in {Appendix~\ref{app:setting}}. We demonstrate effective convergence of our method in Figure~\ref{fig:blo_loss_ag_ar} in Appendix~\ref{app:theory}. Compared with Figure~\ref{fig:warmup}, LIAR helps each individual objective make concerted effort, thus leading to smoother training trajectory. Note that according to~\citet{zhang2024introduction},  the tractability of the convergence of BLO relies on the convexity of the lower-level problems objective of Eq.~\ref{eq:trades}. We thus provide a theoretical proof for the convexity in Appendix~\ref{app:theory}.

% Although without a proper theoretical analysis framework, our method converges well in practice (see \textcolor{blue}{Appendix C}).

\begin{table*}[t]

\centering
% \scriptsize
\resizebox{\linewidth}{!}{
\begin{tabular}{@{}cccccccccccccc@{}}
\toprule
\multicolumn{3}{c}{\textbf{Experiment}} & \multicolumn{5}{c}{\textbf{Harmful Behavior / Target Database}} & \multicolumn{5}{c}{\textbf{Harmful String / Target Database}} \\
\cmidrule(lr){1-3} \cmidrule(lr){4-8} \cmidrule(lr){9-13}
Source Model & Target Model  & Source Database & NQ $\uparrow$ & MS $\uparrow$ & HQ $\uparrow$ & FQ $\uparrow$ & QR $\uparrow$ & NQ $\uparrow$ & MS $\uparrow$ & HQ $\uparrow$ & FQ $\uparrow$ & QR $\uparrow$ \\ \midrule
\multirow{7}{*}{LLaMA-2-7B} & \multirow{2}{*}{LLaMA-2-13B} & NQ & 0.3865 & 0.3596 & 0.3788 & 0.3538 & 0.3635 & 0.3502 & 0.3118 & 0.3502 & 0.3066 & 0.3153 \\
& & MS & 0.3385 & 0.3500 & 0.3404 & 0.3250 & 0.3346 & 0.2927 & 0.3153 & 0.3153 & 0.2857 & 0.2892 \\ \cmidrule(lr){2-13}
& \multirow{2}{*}{Vicuna-13B} & NQ & 0.3788 & 0.3519 & 0.3731 & 0.3481 & 0.3577 & 0.3432 & 0.3066 & 0.3432 & 0.3014 & 0.3101 \\
& & MS & 0.3442 & 0.3558 & 0.3462 & 0.3327 & 0.3404 & 0.2979 & 0.3223 & 0.3223 & 0.2909 & 0.2944 \\ \cmidrule(lr){2-13}
& \multirow{2}{*}{GPT-3.5} & NQ & 0.1904 & 0.1769 & 0.1865 & 0.1750 & 0.1808 & 0.1725 & 0.1533 & 0.1725 & 0.1516 & 0.1568 \\
& & MS & 0.1673 & 0.1712 & 0.1673 & 0.1596 & 0.1654 & 0.1446 & 0.1551 & 0.1551 & 0.1411 & 0.1429 \\ \midrule
\multirow{7}{*}{Vicuna-7B} & \multirow{2}{*}{LLaMA-2-13B} & NQ & 0.3192 & 0.2962 & 0.3135 & 0.2923 & 0.3019 & 0.2857 & 0.2544 & 0.2857 & 0.2509 & 0.2578 \\
& & MS & 0.2808 & 0.2904 & 0.2827 & 0.2712 & 0.2788 & 0.2404 & 0.2596 & 0.2596 & 0.2352 & 0.2387 \\ \cmidrule(lr){2-13}
& \multirow{2}{*}{Vicuna-13B} & NQ & 0.3654 & 0.3385 & 0.3577 & 0.3346 & 0.3442 & 0.3275 & 0.2909 & 0.3275 & 0.2875 & 0.2962 \\
& & MS & 0.3346 & 0.3442 & 0.3346 & 0.3212 & 0.3308 & 0.2857 & 0.3084 & 0.3084 & 0.2787 & 0.2822 \\ \cmidrule(lr){2-13}
& \multirow{2}{*}{GPT-3.5} & NQ & 0.1712 & 0.1596 & 0.1673 & 0.1558 & 0.1615 & 0.1533 & 0.1359 & 0.1533 & 0.1341 & 0.1376 \\
& & MS & 0.1500 & 0.1558 & 0.1500 & 0.1442 & 0.1481 & 0.1289 & 0.1394 & 0.1394 & 0.1254 & 0.1272 \\ \midrule
\multirow{7}{*}{Ensemble} & \multirow{2}{*}{LLaMA-2-13B} & NQ & 0.5500 & 0.4827 & 0.5173 & 0.4769 & 0.4904 & 0.4913 & 0.4146 & 0.4634 & 0.4094 & 0.4199 \\
& & MS & 0.4750 & 0.5192 & 0.4885 & 0.4577 & 0.4692 & 0.4111 & 0.4686 & 0.4425 & 0.4007 & 0.4077 \\ \cmidrule(lr){2-13}
& \multirow{2}{*}{Vicuna-13B} & NQ & 0.5846 & 0.5135 & 0.5500 & 0.5077 & 0.5212 & 0.5226 & 0.4408 & 0.4930 & 0.4355 & 0.4460 \\
& & MS & 0.5231 & 0.5731 & 0.5404 & 0.5058 & 0.5173 & 0.4547 & 0.5174 & 0.4878 & 0.4425 & 0.4495 \\ \cmidrule(lr){2-13}
& \multirow{2}{*}{GPT-3.5} & NQ & 0.2942 & 0.2596 & 0.2769 & 0.2558 & 0.2615 & 0.2631 & 0.2213 & 0.2474 & 0.2195 & 0.2247 \\
& & MS & 0.2519 & 0.2769 & 0.2615 & 0.2442 & 0.2500 & 0.2195 & 0.2509 & 0.2352 & 0.2143 & 0.2178 \\ \bottomrule
\end{tabular}}
\caption{\small
Results of gray-box attack based on LIAR for RAG systems with different knowledge databases and LLM generators. We consider the two adversarial goals defined in \protect Section~\ref{sec:goal} with example case studies in \protect Appendix~\ref{app:case}. Model settings including ensemble are detailed in \protect Appendix~\ref{app:setting}.}
\label{tab:compare}
\vspace{-3mm}
\end{table*}

\vspace{-1mm}
\section{Experiments}\label{sec:exp}
We conduct a series of experiments to evaluate the effectiveness of LIAR. Detailed \textit{Experiment Settings}, including (1) dataset for attacks, (2) knowledge databases, (3) Retriever models, (4) LLM models, and (5) Training details are included in Appendix~\ref{app:setting}.
% The performance of the generated attacks was measured across various metrics, focusing on the impact on both the retrieval and generative components of the system.
% \vspace{-2mm}
\textit{Evaluation Protocol:}
We set the Attack Success Rate (ASR) as the primary metric and evaluate the result by text matching and human judgment akin to~\citet{zou2023universal}.
% \subsection{White-box Attack}
% the main result of ASR

\subsection{Overall Performance of LIAR}

Table~\ref{tab:compare} summarizes the effectiveness of LIAR for gray-box attacks on various RAG systems, with different source and target models and knowledge bases. We obtain the following key observations:
\noindent\ding{182}~\underline{\textbf{Performance Variability:}}
The effectiveness of gray-box attacks varies significantly across different model pairings. For example, when using LLaMA-2-7B as the source model, attacks on LLaMA-2-13B show relatively higher Harmful Behavior rates, such as 0.3865 for NQ and 0.3596 for MS, compared to Vicuna-13B and GPT-3.5 targets. This suggests that attacks are more effective when source and target models are similar.
\noindent\ding{183}~\underline{\textbf{Knowledge Base Sensitivity:}}
Different knowledge bases exhibit varying levels of vulnerability. The NQ and MS databases consistently show higher Harmful Behavior detection rates, such as 0.3865 and 0.3596 for LLaMA-2-13B under attack by LLaMA-2-7B. In contrast, HQ and FQ databases tend to be less impacted, with lower detection rates, highlighting that the nature of the database content influences attack susceptibility.
\noindent\ding{184}~\underline{\textbf{Ensemble Approach Efficacy:}}
Ensemble attacks, which combine multiple models, generally perform better. For instance, attacks on Vicuna-13B using an ensemble approach show a Harmful Behavior rate of 0.5846 for NQ and 0.5135 for MS. This indicates that using multiple models can enhance the transferability of the generated adversarial content attacks.
\noindent\ding{185}~\underline{\textbf{Behavior Detection Rates:}}
Harmful String detection rates are lower than Harmful Behavior rates across the board. For example, the highest string detection for LLaMA-2-13B under attack by LLaMA-2-7B is 0.3502 for NQ, suggesting that broader content manipulation is more achievable than specific string alterations.
% \noindent\ding{186}~\underline{\textbf{Impact of Model Similarity:}}
% Attacks from similar source and target models, such as LLaMA-2-7B to LLaMA-2-13B, tend to be more effective, showing higher detection rates. Conversely, attacks on GPT-3.5 exhibit lower rates, such as 0.1904 for NQ, indicating that model architecture differences may reduce attack effectiveness.
\noindent\ding{187}~\underline{\textbf{General Observations:}}
The results highlight that adversarial contents learned through vulnerabilities can effectively manipulate RAG systems under the gray-box attack scenario. The vulnerabilities is influenced by the choice of models and knowledge bases. More detailed analyses on each components are explored in following subsections.

\begin{figure*}[t]
    \centering
    \begin{subfigure}{0.32\textwidth}
        \centering
        \includegraphics[width=\textwidth]{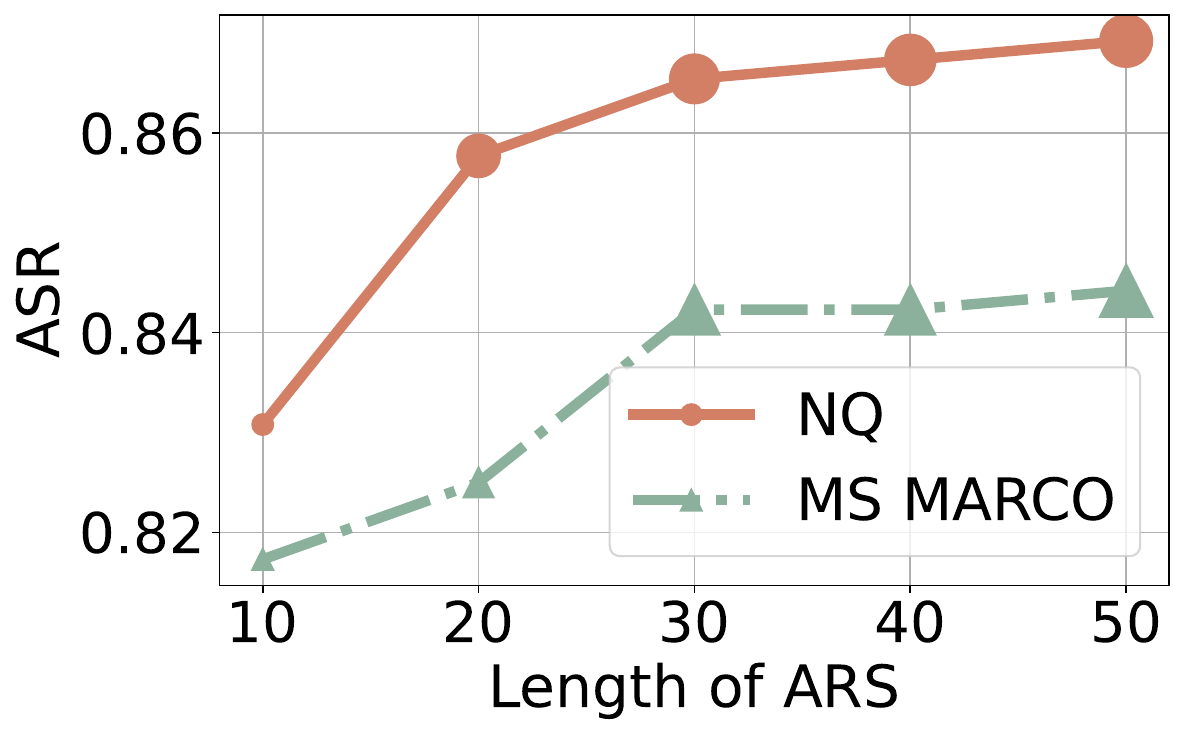}
        \caption{\small ASR \textit{v.s.} Length of ARS.}
        \label{fig:sub1}
    \end{subfigure}
    \hfill
    \begin{subfigure}{0.32\textwidth}
        \centering
        \includegraphics[width=\textwidth]{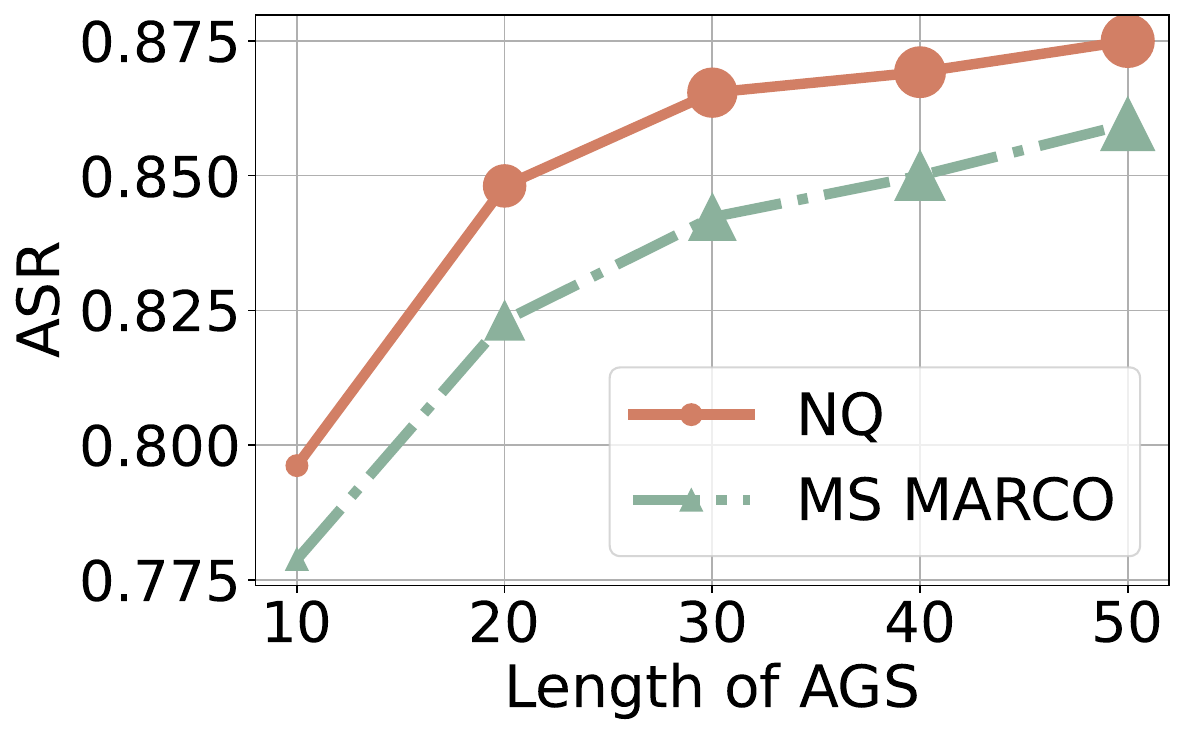}
        \caption{\small ASR \textit{v.s.} Length of AGS.}
        \label{fig:sub2}
    \end{subfigure}
    \hfill
    \begin{subfigure}{0.32\textwidth}
        \centering
        \includegraphics[width=\textwidth]{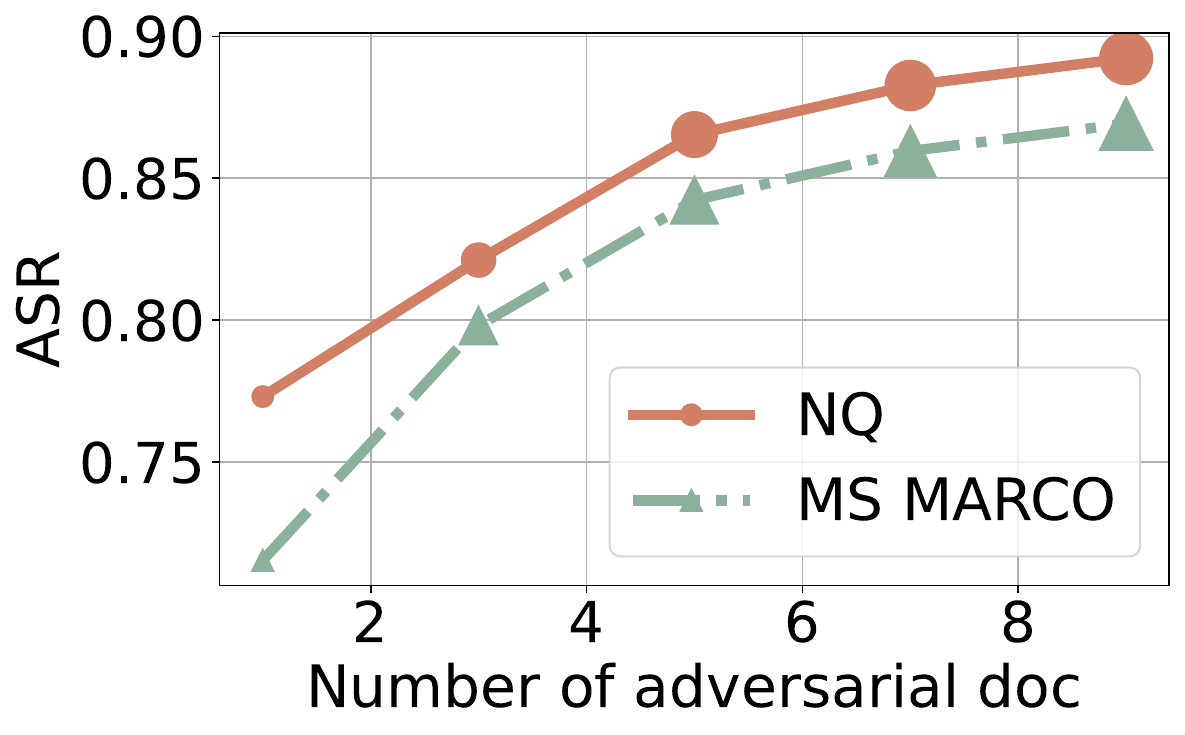}
        \caption{\small ASR \textit{v.s.} Number of Adversarial Doc.}
        \label{fig:sub3}
    \end{subfigure}
    \caption{\small Sensitivity analyses on three key hyper-parameters.}
    \label{fig:sensi}
    % \vspace{-1mm}
\end{figure*}

\subsection{Ablation Study}

In the ablation study, we individually investigate the transferability of the two attack components to assess their effectiveness in different scenarios.

\vspace{-2mm}
\paragraph{Transferability to Unseen Knowledge Database.}
We evaluated the performance of our attack on the retriever when applied to RAG with unseen knowledge database. The transferability is measured by the retrieval success rate of adversarial content across various target databases, as shown in Table~\ref{tab:transfer_databases}. The results indicate that the attack maintains a performance with a success rate exceeding 70\% across different databases. Notably, when transferring to HotpotQA, the attack achieved a success rate of 77.12\%, suggesting robust generalization to diverse question types. However, the performance on FiQA and Quora was slightly lower, highlighting some variability in effectiveness depending on the nature of the queries.

\begin{table}[htpb]
\centering
\scalebox{0.8}{
\begin{tabular}{@{}ccc@{}}
\toprule
\textbf{Target Database} & \textbf{NQ} & \textbf{MS MARCO} \\ \midrule
NQ & NA & 0.7269 \\
MS MARCO & 0.7173 & NA \\
HotpotQA & 0.7712 & 0.7519 \\
FiQA & 0.7077 & 0.7000 \\
Quora & 0.7269 & 0.7192 \\ \bottomrule
\end{tabular}}
\caption{\small Transfer results across different databases}
\label{tab:transfer_databases}
\vspace{-6mm}
\end{table}

% \vspace{-3mm}
\paragraph{Transferability to Unseen LLM Generators.}
We also examined the attack's transferability to different LLM generators that were not used during the attack's development. As depicted in Table~\ref{tab:transfer_models}, the attack was particularly effective when transferred to models with similar architectures to those used in training. For instance, Vicuna-13B showed a high success rate of 58.46\% on NQ and 57.31\% on MS MARCO. In contrast, models like Claude-3-Haiku and Gemini-1.0-Pro exhibited significantly lower transferability rates, with success rates dropping below 3\% for Claude-3-Haiku. These results suggest that the effectiveness of the attack may vary considerably with different model architectures.

\begin{table}[htpb]
\centering
\scalebox{0.75}{
\begin{tabular}{@{}ccc@{}}
\toprule
\textbf{Target Model} & \textbf{NQ} & \textbf{MS MARCO} \\ \midrule
LLaMA-2-13B & 0.5500 & 0.5192 \\
Vicuna-13B & 0.5846 & 0.5731 \\
Claude-3-Haiku & 0.0288 & 0.0212 \\
Gemini-1.0-Pro & 0.2635 & 0.2250 \\
GPT-3.5 & 0.2942 & 0.2769 \\
GPT-4 & 0.1673 & 0.1442 \\ \bottomrule
\end{tabular}}
\caption{\small Transfer results across different models}
\label{tab:transfer_models}
\vspace{-1mm}
\end{table}

\vspace{-3mm}
\paragraph{Impact of Different Attack Components.}
Table~\ref{tab:settings_analysis} presents AR, AG, and ASR for various settings. LIAR shows the highest ASR for both NQ (0.7654) and MS MARCO (0.7288), indicating its effectiveness. The absence of a retriever attack significantly reduces AR and ASR, showing the importance of this component. Notably, the removal of the jailbreak prompt results in an ASR of 0.0000 for both datasets, suggesting its vital role in successful attacks.

\begin{table}[]
\centering
\small
\resizebox{\linewidth}{!}{
\begin{tabular}{@{}cccccc@{}}
\toprule
\textbf{Database} & \textbf{Setting} & \textbf{AR} & \textbf{AG} & \textbf{ASR} \\ \midrule
\multirow{4}{*}{NQ} & w/o retriever attack & 0.0412 & 0.9288 & 0.0135 \\
 & w/o jailbreak prompt & 0.9148 & 0.0000 & 0.0000 \\
 & warm-up training & 0.0703 & 0.0462 & 0.0462 \\
 & LIAR & 0.8740 & 0.7654 & 0.7654 \\ \midrule
\multirow{4}{*}{MS MARCO} & w/o retriever attack & 0.0124 & 0.9288 & 0.0038 \\
 & w/o jailbreak prompt & 0.8672 & 0.0000 & 0.0000 \\
 & warm-up training & 0.0539 & 0.0365 & 0.0365 \\
 & LIAR & 0.8247 & 0.7288 & 0.7288 \\ \bottomrule
\end{tabular}}
\caption{\small AR, AG, and ASR for Different Settings}
\label{tab:settings_analysis}
\vspace{-5mm}
\end{table}

% \vspace{-4mm}
% \paragraph{Summary of Findings.}
% The ablation study shows that our attack is effective on unseen queries and similar LLMs, though performance varies across different models. This indicates a potential for broad generalization with some need for model-specific adjustments to improve robustness.

\vspace{-1mm}
\subsection{Senstivity of Hyper-parameters}
\label{sec:hyperparameters}

Figure~\ref{fig:sensi} shows the impact of varying three parameters on ASR for NQ and MS MARCO datasets. We use LLaMA-2-7B as the LLM generator. 
% We have the following discoveries:

\noindent\ding{182}~\underline{{\textbf{Length of ARS}}} (Figure~\ref{fig:sensi}a).
Increasing ARS length from 10 to 50 tokens slightly improves ASR, with NQ seeing a more noticeable increase from 0.82 to 0.86 compared to MS MARCO, which improves from 0.82 to 0.84.
\noindent\ding{183}~\underline{\textbf{Length of AGS}} (Figure~\ref{fig:sensi}b).
Extending AGS from 10 to 50 tokens also enhances ASR. NQ shows an increase from 0.80 to 0.875, while MS MARCO improves from 0.775 to 0.85, indicating a positive but moderate effect.
\noindent\ding{184}~\underline{\textbf{Number of Adversarial Documents}} (Figure~\ref{fig:sensi}c).
Adding more adversarial documents from 2 to 10 leads to a significant rise in ASR, with NQ increasing from 0.75 to 0.90 and MS MARCO from 0.75 to 0.85, suggesting higher content volume can aid attack success.

Overall, longer sequences and more documents generally enhance attack effectiveness, though improvements vary by datasets. We further provide experiment results in Appendix~\ref{app:more_exp}, including the effetiveness of different retriever models, and effectiveness against classic defense. Case studies can be found in Appendix~\ref{app:case}.

% \vspace{-1mm}
\section{Conclusion}
% \vspace{-1mm}
In this paper, we demonstrated the vulnerabilities of Retrieval-Augmented Generative (RAG) models to gray-box attacks. Through a series of experiments, we showed that adversarial content could significantly impact the retrieval and generative components of these systems. Our findings show the need for robust defense mechanisms to protect against such attacks, ensuring the integrity and reliability of RAG models in various applications. In broader terms, we emphasize the urgent need to strengthen trustworthiness of LLM applications.

\section*{Limitation Discussions \& Future Work}
Despite the promising results, our study has several limitations that warrant discussion. 

Firstly, the scope of our experiments was limited to specific datasets and models, which may not fully capture the diversity and complexity of real-world RAG systems. Future work should extend these evaluations to a broader range of datasets, tasks and models, such as math~\cite{wu2024conceptmath,zhang2024mario}, or even multi-modal scenarios~\cite{feng2022beyond}.

Secondly, our gray-box attack assumes partial knowledge of the retriever, which may not always reflect practical attack scenarios where attackers have less information. 
% Investigating the impact of these attacks under more restrictive conditions or with different levels of adversary knowledge could provide deeper insights into the robustness of RAG models.

Thirdly, while we demonstrated the effectiveness of our attack in controlled settings, the real-world applicability and impact need further exploration. Real-world systems often involve additional complexities such as continuous updates and dynamic content changes, which were not accounted for in our static evaluation framework. Future work should focus on developing adaptive attack strategies that can cope with these dynamics.

Moreover, our approach primarily targets the text-based RAG systems, and its applicability to multimodal RAG systems, which integrate text with other data forms such as images or audio, remains unexplored. Expanding our methodology to address multimodal contexts will be an important area of future research.

Lastly, our work highlights the need for robust defense mechanisms against adversarial attacks. Future research should aim to develop and evaluate more effective defense strategies, including adversarial training and anomaly detection techniques, to enhance the resilience of RAG models against such threats.

\section*{Ethical Statement}
Our research on attacking RAG models aims to highlight and address potential security vulnerabilities in AI systems. The intention behind this study is to raise awareness about the risks associated with the use of RAG models and to promote the development of more secure and reliable AI technologies.

We acknowledge that the techniques discussed could potentially be misused to cause harm or manipulate information. To mitigate these risks, our work adheres to the principles of responsible disclosure, ensuring that the details provided are sufficient for researchers and practitioners to understand and counteract the vulnerabilities without enabling malicious use. We strongly advocate for the responsible application of AI technologies and emphasize that the findings from this study should be used solely for improving system security.

Additionally, we conducted our experiments in a controlled environment and did not involve real user data or deploy any harmful actions that could affect individuals or organizations. We are committed to ensuring that our research practices align with ethical guidelines and contribute positively to the field of AI security.

% \begin{ack}
% Use unnumbered first level headings for the acknowledgments. All acknowledgments
% More information about this disclosure can be found at: \url{https://neurips.cc/Conferences/2024/PaperInformation/FundingDisclosure}.

% Do {\bf not} include this section in the anonymized submission, only in the final paper. You can use the \texttt{ack} environment provided in the style file to automatically hide this section in the anonymized submission.
% \end{ack}

\bibliography{ref}

\begin{thebibliography}{59}
\providecommand{\natexlab}[1]{#1}

\bibitem[{Anderson et~al.(2024)Anderson, Amit, and Goldsteen}]{anderson2024my}
Maya Anderson, Guy Amit, and Abigail Goldsteen. 2024.
\newblock Is my data in your retrieval database? membership inference attacks against retrieval augmented generation.
\newblock \emph{arXiv preprint arXiv:2405.20446}.

\bibitem[{Anthropic(2024)}]{anthropic2024claude}
AI~Anthropic. 2024.
\newblock The claude 3 model family: Opus, sonnet, haiku.
\newblock \emph{Claude-3 Model Card}.

\bibitem[{Bai et~al.(2022)Bai, Kadavath, Kundu, Askell, Kernion, Jones, Chen, Goldie, Mirhoseini, McKinnon et~al.}]{bai2022constitutional}
Yuntao Bai, Saurav Kadavath, Sandipan Kundu, Amanda Askell, Jackson Kernion, Andy Jones, Anna Chen, Anna Goldie, Azalia Mirhoseini, Cameron McKinnon, et~al. 2022.
\newblock Constitutional ai: Harmlessness from ai feedback.
\newblock \emph{arXiv preprint arXiv:2212.08073}.

\bibitem[{Bezdek and Hathaway(2003)}]{bezdek2003convergence}
James~C Bezdek and Richard~J Hathaway. 2003.
\newblock Convergence of alternating optimization.
\newblock \emph{Neural, Parallel \& Scientific Computations}, 11(4):351--368.

\bibitem[{Brown et~al.(2020)Brown, Mann, Ryder, Subbiah, Kaplan, Dhariwal, Neelakantan, Shyam, Sastry, Askell, Agarwal, Herbert{-}Voss, Krueger, Henighan, Child, Ramesh, Ziegler, Wu, Winter, Hesse, Chen, Sigler, Litwin, Gray, Chess, Clark, Berner, McCandlish, Radford, Sutskever, and Amodei}]{DBLP:conf/nips/BrownMRSKDNSSAA20}
Tom~B. Brown, Benjamin Mann, Nick Ryder, Melanie Subbiah, Jared Kaplan, Prafulla Dhariwal, Arvind Neelakantan, Pranav Shyam, Girish Sastry, Amanda Askell, Sandhini Agarwal, Ariel Herbert{-}Voss, Gretchen Krueger, Tom Henighan, Rewon Child, Aditya Ramesh, Daniel~M. Ziegler, Jeffrey Wu, Clemens Winter, Christopher Hesse, Mark Chen, Eric Sigler, Mateusz Litwin, Scott Gray, Benjamin Chess, Jack Clark, Christopher Berner, Sam McCandlish, Alec Radford, Ilya Sutskever, and Dario Amodei. 2020.
\newblock Language models are few-shot learners.
\newblock In \emph{NeurIPS}.

\bibitem[{Chen et~al.(2024)Chen, Lin, Han, and Sun}]{chen2024benchmarking}
Jiawei Chen, Hongyu Lin, Xianpei Han, and Le~Sun. 2024.
\newblock Benchmarking large language models in retrieval-augmented generation.
\newblock In \emph{Proceedings of the AAAI Conference on Artificial Intelligence}, volume~38, pages 17754--17762.

\bibitem[{Cheng et~al.(2024)Cheng, Ding, Ju, Wu, Du, Yi, Zhang, and Liu}]{cheng2024trojanrag}
Pengzhou Cheng, Yidong Ding, Tianjie Ju, Zongru Wu, Wei Du, Ping Yi, Zhuosheng Zhang, and Gongshen Liu. 2024.
\newblock Trojanrag: Retrieval-augmented generation can be backdoor driver in large language models.
\newblock \emph{arXiv preprint arXiv:2405.13401}.

\bibitem[{Chiang et~al.(2023)Chiang, Li, Lin, Sheng, Wu, Zhang, Zheng, Zhuang, Zhuang, Gonzalez, Stoica, and Xing}]{vicuna2023}
Wei-Lin Chiang, Zhuohan Li, Zi~Lin, Ying Sheng, Zhanghao Wu, Hao Zhang, Lianmin Zheng, Siyuan Zhuang, Yonghao Zhuang, Joseph~E. Gonzalez, Ion Stoica, and Eric~P. Xing. 2023.
\newblock \href {https://lmsys.org/blog/2023-03-30-vicuna/} {Vicuna: An open-source chatbot impressing gpt-4 with 90\%* chatgpt quality}.

\bibitem[{Cho et~al.(2024)Cho, Jeong, Seo, Hwang, and Park}]{cho2024typos}
Sukmin Cho, Soyeong Jeong, Jeongyeon Seo, Taeho Hwang, and Jong~C Park. 2024.
\newblock Typos that broke the rag's back: Genetic attack on rag pipeline by simulating documents in the wild via low-level perturbations.
\newblock \emph{arXiv preprint arXiv:2404.13948}.

\bibitem[{Dettmers et~al.(2023)Dettmers, Pagnoni, Holtzman, and Zettlemoyer}]{DBLP:conf/nips/DettmersPHZ23}
Tim Dettmers, Artidoro Pagnoni, Ari Holtzman, and Luke Zettlemoyer. 2023.
\newblock Qlora: Efficient finetuning of quantized llms.
\newblock In \emph{NeurIPS}.

\bibitem[{Ebrahimi et~al.(2017)Ebrahimi, Rao, Lowd, and Dou}]{ebrahimi2017hotflip}
Javid Ebrahimi, Anyi Rao, Daniel Lowd, and Dejing Dou. 2017.
\newblock Hotflip: White-box adversarial examples for text classification.
\newblock \emph{arXiv preprint arXiv:1712.06751}.

\bibitem[{Feng et~al.(2022)Feng, Bu, Zhang, and Li}]{feng2022beyond}
Weixin Feng, Xingyuan Bu, Chenchen Zhang, and Xubin Li. 2022.
\newblock Beyond bounding box: Multimodal knowledge learning for object detection.
\newblock \emph{arXiv preprint arXiv:2205.04072}.

\bibitem[{Finn et~al.(2017)Finn, Abbeel, and Levine}]{finn2017model}
Chelsea Finn, Pieter Abbeel, and Sergey Levine. 2017.
\newblock Model-agnostic meta-learning for fast adaptation of deep networks.
\newblock In \emph{International conference on machine learning}, pages 1126--1135. PMLR.

\bibitem[{Gao et~al.(2023)Gao, Xiong, Gao, Jia, Pan, Bi, Dai, Sun, and Wang}]{gao2023retrieval}
Yunfan Gao, Yun Xiong, Xinyu Gao, Kangxiang Jia, Jinliu Pan, Yuxi Bi, Yi~Dai, Jiawei Sun, and Haofen Wang. 2023.
\newblock Retrieval-augmented generation for large language models: A survey.
\newblock \emph{arXiv preprint arXiv:2312.10997}.

\bibitem[{Guo et~al.(2021)Guo, Sablayrolles, J{\'e}gou, and Kiela}]{guo2021gradient}
Chuan Guo, Alexandre Sablayrolles, Herv{\'e} J{\'e}gou, and Douwe Kiela. 2021.
\newblock Gradient-based adversarial attacks against text transformers.
\newblock \emph{arXiv preprint arXiv:2104.13733}.

\bibitem[{Heidi~Steen(2024)}]{bing24}
Dan~Wahlin Heidi~Steen. 2024.
\newblock \href {https://learn.microsoft.com/en-us/azure/search/retrieval-augmented-generation-overview} {Retrieval augmented generation (rag) in azure ai search}.

\bibitem[{Izacard et~al.(2021)Izacard, Caron, Hosseini, Riedel, Bojanowski, Joulin, and Grave}]{izacard2021unsupervised}
Gautier Izacard, Mathilde Caron, Lucas Hosseini, Sebastian Riedel, Piotr Bojanowski, Armand Joulin, and Edouard Grave. 2021.
\newblock Unsupervised dense information retrieval with contrastive learning.
\newblock \emph{arXiv preprint arXiv:2112.09118}.

\bibitem[{Izacard et~al.(2022)Izacard, Caron, Hosseini, Riedel, Bojanowski, Joulin, and Grave}]{DBLP:journals/tmlr/IzacardCHRBJG22}
Gautier Izacard, Mathilde Caron, Lucas Hosseini, Sebastian Riedel, Piotr Bojanowski, Armand Joulin, and Edouard Grave. 2022.
\newblock Unsupervised dense information retrieval with contrastive learning.
\newblock \emph{Trans. Mach. Learn. Res.}, 2022.

\bibitem[{Jang et~al.(2016)Jang, Gu, and Poole}]{jang2016categorical}
Eric Jang, Shixiang Gu, and Ben Poole. 2016.
\newblock Categorical reparameterization with gumbel-softmax.
\newblock \emph{arXiv preprint arXiv:1611.01144}.

\bibitem[{Joo et~al.(2020)Joo, Kim, Shin, and Moon}]{joo2020generalized}
Weonyoung Joo, Dongjun Kim, Seungjae Shin, and Il-Chul Moon. 2020.
\newblock Generalized gumbel-softmax gradient estimator for various discrete random variables.
\newblock \emph{arXiv preprint arXiv:2003.01847}.

\bibitem[{Kaz~Sato(2024)}]{gg24}
Guangsha~Shi Kaz~Sato. 2024.
\newblock \href {https://cloud.google.com/blog/products/ai-machine-learning/rags-powered-by-google-search-technology-part-1} {Your rags powered by google search technology}.

\bibitem[{Khaliq et~al.(2024)Khaliq, Chang, Ma, Pflugfelder, and Mileti{\'c}}]{khaliq2024ragar}
M~Abdul Khaliq, P~Chang, M~Ma, B~Pflugfelder, and F~Mileti{\'c}. 2024.
\newblock Ragar, your falsehood radar: Rag-augmented reasoning for political fact-checking using multimodal large language models.
\newblock \emph{arXiv preprint arXiv:2404.12065}.

\bibitem[{Komeili et~al.(2021)Komeili, Shuster, and Weston}]{komeili2021internet}
Mojtaba Komeili, Kurt Shuster, and Jason Weston. 2021.
\newblock Internet-augmented dialogue generation.
\newblock \emph{arXiv preprint arXiv:2107.07566}.

\bibitem[{Kwiatkowski et~al.(2019)Kwiatkowski, Palomaki, Redfield, Collins, Parikh, Alberti, Epstein, Polosukhin, Devlin, Lee, Toutanova, Jones, Kelcey, Chang, Dai, Uszkoreit, Le, and Petrov}]{DBLP:journals/tacl/KwiatkowskiPRCP19}
Tom Kwiatkowski, Jennimaria Palomaki, Olivia Redfield, Michael Collins, Ankur~P. Parikh, Chris Alberti, Danielle Epstein, Illia Polosukhin, Jacob Devlin, Kenton Lee, Kristina Toutanova, Llion Jones, Matthew Kelcey, Ming{-}Wei Chang, Andrew~M. Dai, Jakob Uszkoreit, Quoc Le, and Slav Petrov. 2019.
\newblock \href {https://doi.org/10.1162/TACL\_A\_00276} {Natural questions: a benchmark for question answering research}.
\newblock \emph{Trans. Assoc. Comput. Linguistics}, 7:452--466.

\bibitem[{Lewis et~al.(2020)Lewis, Perez, Piktus, Petroni, Karpukhin, Goyal, K{\"u}ttler, Lewis, Yih, Rockt{\"a}schel et~al.}]{lewis2020retrieval}
Patrick Lewis, Ethan Perez, Aleksandra Piktus, Fabio Petroni, Vladimir Karpukhin, Naman Goyal, Heinrich K{\"u}ttler, Mike Lewis, Wen-tau Yih, Tim Rockt{\"a}schel, et~al. 2020.
\newblock Retrieval-augmented generation for knowledge-intensive nlp tasks.
\newblock \emph{Advances in Neural Information Processing Systems}, 33:9459--9474.

\bibitem[{Li et~al.(2022)Li, Su, Cai, Wang, and Liu}]{li2022survey}
Huayang Li, Yixuan Su, Deng Cai, Yan Wang, and Lemao Liu. 2022.
\newblock A survey on retrieval-augmented text generation.
\newblock \emph{arXiv preprint arXiv:2202.01110}.

\bibitem[{Li et~al.(2024)Li, Jin, Zhou, Zhang, Zhang, Zhu, and Dou}]{li2024matching}
Xiaoxi Li, Jiajie Jin, Yujia Zhou, Yuyao Zhang, Peitian Zhang, Yutao Zhu, and Zhicheng Dou. 2024.
\newblock From matching to generation: A survey on generative information retrieval.
\newblock \emph{arXiv preprint arXiv:2404.14851}.

\bibitem[{Liu et~al.(2021)Liu, Gao, Zhang, Meng, and Lin}]{liu2021investigating}
Risheng Liu, Jiaxin Gao, Jin Zhang, Deyu Meng, and Zhouchen Lin. 2021.
\newblock Investigating bi-level optimization for learning and vision from a unified perspective: A survey and beyond.
\newblock \emph{IEEE Transactions on Pattern Analysis and Machine Intelligence}, 44(12):10045--10067.

\bibitem[{Liu et~al.(2023{\natexlab{a}})Liu, Xu, Chen, and Xiao}]{liu2023autodan}
Xiaogeng Liu, Nan Xu, Muhao Chen, and Chaowei Xiao. 2023{\natexlab{a}}.
\newblock Autodan: Generating stealthy jailbreak prompts on aligned large language models.
\newblock \emph{arXiv preprint arXiv:2310.04451}.

\bibitem[{Liu et~al.(2023{\natexlab{b}})Liu, Deng, Xu, Li, Zheng, Zhang, Zhao, Zhang, and Liu}]{liu2023jailbreaking}
Yi~Liu, Gelei Deng, Zhengzi Xu, Yuekang Li, Yaowen Zheng, Ying Zhang, Lida Zhao, Tianwei Zhang, and Yang Liu. 2023{\natexlab{b}}.
\newblock Jailbreaking chatgpt via prompt engineering: An empirical study.
\newblock \emph{arXiv preprint arXiv:2305.13860}.

\bibitem[{Liu et~al.(2023{\natexlab{c}})Liu, Zhang, Guo, de~Rijke, Chen, Fan, and Cheng}]{liu2023black}
Yu-An Liu, Ruqing Zhang, Jiafeng Guo, Maarten de~Rijke, Wei Chen, Yixing Fan, and Xueqi Cheng. 2023{\natexlab{c}}.
\newblock Black-box adversarial attacks against dense retrieval models: A multi-view contrastive learning method.
\newblock In \emph{Proceedings of the 32nd ACM International Conference on Information and Knowledge Management}, pages 1647--1656.

\bibitem[{Lyu et~al.(2024)Lyu, Lin, Zheng, Pang, Ling, Jha, and Chen}]{lyu2024task}
Weimin Lyu, Xiao Lin, Songzhu Zheng, Lu~Pang, Haibin Ling, Susmit Jha, and Chao Chen. 2024.
\newblock Task-agnostic detector for insertion-based backdoor attacks.
\newblock \emph{arXiv preprint arXiv:2403.17155}.

\bibitem[{Lyu et~al.(2022)Lyu, Zheng, Ma, and Chen}]{lyu2022study}
Weimin Lyu, Songzhu Zheng, Tengfei Ma, and Chao Chen. 2022.
\newblock A study of the attention abnormality in trojaned berts.
\newblock In \emph{Proceedings of the 2022 Conference of the North American Chapter of the Association for Computational Linguistics: Human Language Technologies}, pages 4727--4741.

\bibitem[{Lyu et~al.(2023)Lyu, Zheng, Pang, Ling, and Chen}]{lyu2023attention}
Weimin Lyu, Songzhu Zheng, Lu~Pang, Haibin Ling, and Chao Chen. 2023.
\newblock Attention-enhancing backdoor attacks against bert-based models.
\newblock In \emph{Findings of the Association for Computational Linguistics: EMNLP 2023}, pages 10672--10690.

\bibitem[{Maia et~al.(2018)Maia, Handschuh, Freitas, Davis, McDermott, Zarrouk, and Balahur}]{DBLP:conf/www/MaiaHFDMZB18}
Macedo Maia, Siegfried Handschuh, Andr{\'{e}} Freitas, Brian Davis, Ross McDermott, Manel Zarrouk, and Alexandra Balahur. 2018.
\newblock Www'18 open challenge: Financial opinion mining and question answering.
\newblock In \emph{{WWW} (Companion Volume)}, pages 1941--1942. {ACM}.

\bibitem[{Nguyen et~al.(2016)Nguyen, Rosenberg, Song, Gao, Tiwary, Majumder, and Deng}]{DBLP:conf/nips/NguyenRSGTMD16}
Tri Nguyen, Mir Rosenberg, Xia Song, Jianfeng Gao, Saurabh Tiwary, Rangan Majumder, and Li~Deng. 2016.
\newblock {MS} {MARCO:} {A} human generated machine reading comprehension dataset.
\newblock In \emph{CoCo@NIPS}, volume 1773 of \emph{{CEUR} Workshop Proceedings}. CEUR-WS.org.

\bibitem[{OpenAI(2023)}]{DBLP:journals/corr/abs-2303-08774}
OpenAI. 2023.
\newblock {GPT-4} technical report.
\newblock \emph{CoRR}, abs/2303.08774.

\bibitem[{Qi et~al.(2024)Qi, Huang, Panda, Henderson, Wang, and Mittal}]{qi2024visual}
Xiangyu Qi, Kaixuan Huang, Ashwinee Panda, Peter Henderson, Mengdi Wang, and Prateek Mittal. 2024.
\newblock Visual adversarial examples jailbreak aligned large language models.
\newblock In \emph{Proceedings of the AAAI Conference on Artificial Intelligence}, volume~38, pages 21527--21536.

\bibitem[{Raval and Verma(2020)}]{raval2020one}
Nisarg Raval and Manisha Verma. 2020.
\newblock One word at a time: adversarial attacks on retrieval models.
\newblock \emph{arXiv preprint arXiv:2008.02197}.

\bibitem[{Song et~al.(2020)Song, Rush, and Shmatikov}]{song2020adversarial}
Congzheng Song, Alexander~M Rush, and Vitaly Shmatikov. 2020.
\newblock Adversarial semantic collisions.
\newblock \emph{arXiv preprint arXiv:2011.04743}.

\bibitem[{Song et~al.(2022)Song, Zhang, Zhu, Tang, and Yang}]{song2022trattack}
Junshuai Song, Jiangshan Zhang, Jifeng Zhu, Mengyun Tang, and Yong Yang. 2022.
\newblock Trattack: Text rewriting attack against text retrieval.
\newblock In \emph{Proceedings of the 7th Workshop on Representation Learning for NLP}, pages 191--203.

\bibitem[{Tan et~al.(2024)Tan, Zhao, Moraffah, Li, Kong, Chen, and Liu}]{tan2024wolf}
Zhen Tan, Chengshuai Zhao, Raha Moraffah, Yifan Li, Yu~Kong, Tianlong Chen, and Huan Liu. 2024.
\newblock The wolf within: Covert injection of malice into mllm societies via an mllm operative.
\newblock \emph{arXiv preprint arXiv:2402.14859}.

\bibitem[{Team et~al.(2023)Team, Anil, Borgeaud, Wu, Alayrac, Yu, Soricut, Schalkwyk, Dai, Hauth et~al.}]{team2023gemini}
Gemini Team, Rohan Anil, Sebastian Borgeaud, Yonghui Wu, Jean-Baptiste Alayrac, Jiahui Yu, Radu Soricut, Johan Schalkwyk, Andrew~M Dai, Anja Hauth, et~al. 2023.
\newblock Gemini: a family of highly capable multimodal models.
\newblock \emph{arXiv preprint arXiv:2312.11805}.

\bibitem[{Thakur et~al.(2021{\natexlab{a}})Thakur, Reimers, R{\"{u}}ckl{\'{e}}, Srivastava, and Gurevych}]{DBLP:conf/nips/Thakur0RSG21}
Nandan Thakur, Nils Reimers, Andreas R{\"{u}}ckl{\'{e}}, Abhishek Srivastava, and Iryna Gurevych. 2021{\natexlab{a}}.
\newblock \href {https://datasets-benchmarks-proceedings.neurips.cc/paper/2021/hash/65b9eea6e1cc6bb9f0cd2a47751a186f-Abstract-round2.html} {{BEIR:} {A} heterogeneous benchmark for zero-shot evaluation of information retrieval models}.
\newblock In \emph{Proceedings of the Neural Information Processing Systems Track on Datasets and Benchmarks 1, NeurIPS Datasets and Benchmarks 2021, December 2021, virtual}.

\bibitem[{Thakur et~al.(2021{\natexlab{b}})Thakur, Reimers, R{\"u}ckl{\'e}, Srivastava, and Gurevych}]{thakur2021beir}
Nandan Thakur, Nils Reimers, Andreas R{\"u}ckl{\'e}, Abhishek Srivastava, and Iryna Gurevych. 2021{\natexlab{b}}.
\newblock Beir: A heterogenous benchmark for zero-shot evaluation of information retrieval models.
\newblock \emph{arXiv preprint arXiv:2104.08663}.

\bibitem[{Touvron et~al.(2023)Touvron, Martin, Stone, Albert, Almahairi, Babaei, Bashlykov, Batra, Bhargava, Bhosale, Bikel, Blecher, Canton{-}Ferrer, Chen, Cucurull, Esiobu, Fernandes, Fu, Fu, Fuller, Gao, Goswami, Goyal, Hartshorn, Hosseini, Hou, Inan, Kardas, Kerkez, Khabsa, Kloumann, Korenev, Koura, Lachaux, Lavril, Lee, Liskovich, Lu, Mao, Martinet, Mihaylov, Mishra, Molybog, Nie, Poulton, Reizenstein, Rungta, Saladi, Schelten, Silva, Smith, Subramanian, Tan, Tang, Taylor, Williams, Kuan, Xu, Yan, Zarov, Zhang, Fan, Kambadur, Narang, Rodriguez, Stojnic, Edunov, and Scialom}]{DBLP:journals/corr/abs-2307-09288}
Hugo Touvron, Louis Martin, Kevin Stone, Peter Albert, Amjad Almahairi, Yasmine Babaei, Nikolay Bashlykov, Soumya Batra, Prajjwal Bhargava, Shruti Bhosale, Dan Bikel, Lukas Blecher, Cristian Canton{-}Ferrer, Moya Chen, Guillem Cucurull, David Esiobu, Jude Fernandes, Jeremy Fu, Wenyin Fu, Brian Fuller, Cynthia Gao, Vedanuj Goswami, Naman Goyal, Anthony Hartshorn, Saghar Hosseini, Rui Hou, Hakan Inan, Marcin Kardas, Viktor Kerkez, Madian Khabsa, Isabel Kloumann, Artem Korenev, Punit~Singh Koura, Marie{-}Anne Lachaux, Thibaut Lavril, Jenya Lee, Diana Liskovich, Yinghai Lu, Yuning Mao, Xavier Martinet, Todor Mihaylov, Pushkar Mishra, Igor Molybog, Yixin Nie, Andrew Poulton, Jeremy Reizenstein, Rashi Rungta, Kalyan Saladi, Alan Schelten, Ruan Silva, Eric~Michael Smith, Ranjan Subramanian, Xiaoqing~Ellen Tan, Binh Tang, Ross Taylor, Adina Williams, Jian~Xiang Kuan, Puxin Xu, Zheng Yan, Iliyan Zarov, Yuchen Zhang, Angela Fan, Melanie Kambadur, Sharan Narang, Aur{\'{e}}lien Rodriguez, Robert Stojnic, Sergey Edunov,
  and Thomas Scialom. 2023.
\newblock Llama 2: Open foundation and fine-tuned chat models.
\newblock \emph{CoRR}, abs/2307.09288.

\bibitem[{Wang et~al.(2024)Wang, Khramtsova, Zhuang, and Zuccon}]{wang2024feb4rag}
Shuai Wang, Ekaterina Khramtsova, Shengyao Zhuang, and Guido Zuccon. 2024.
\newblock Feb4rag: Evaluating federated search in the context of retrieval augmented generation.
\newblock \emph{arXiv preprint arXiv:2402.11891}.

\bibitem[{Wei et~al.(2023)Wei, Haghtalab, and Steinhardt}]{NEURIPS2023_fd661313}
Alexander Wei, Nika Haghtalab, and Jacob Steinhardt. 2023.
\newblock \href {https://proceedings.neurips.cc/paper_files/paper/2023/file/fd6613131889a4b656206c50a8bd7790-Paper-Conference.pdf} {Jailbroken: How does llm safety training fail?}
\newblock In \emph{Advances in Neural Information Processing Systems}, volume~36, pages 80079--80110. Curran Associates, Inc.

\bibitem[{Wei et~al.(2024)Wei, Yang, Song, Lu, Hu, Tran, Peng, Liu, Huang, Du et~al.}]{wei2024long}
Jerry Wei, Chengrun Yang, Xinying Song, Yifeng Lu, Nathan Hu, Dustin Tran, Daiyi Peng, Ruibo Liu, Da~Huang, Cosmo Du, et~al. 2024.
\newblock Long-form factuality in large language models.
\newblock \emph{arXiv preprint arXiv:2403.18802}.

\bibitem[{Wu et~al.(2024)Wu, Liu, Bu, Liu, Zhou, Zhang, Zhang, Bai, Chen, Ge et~al.}]{wu2024conceptmath}
Yanan Wu, Jie Liu, Xingyuan Bu, Jiaheng Liu, Zhanhui Zhou, Yuanxing Zhang, Chenchen Zhang, Zhiqi Bai, Haibin Chen, Tiezheng Ge, et~al. 2024.
\newblock Conceptmath: A bilingual concept-wise benchmark for measuring mathematical reasoning of large language models.
\newblock \emph{arXiv preprint arXiv:2402.14660}.

\bibitem[{Xiong et~al.(2021)Xiong, Xiong, Li, Tang, Liu, Bennett, Ahmed, and Overwijk}]{DBLP:conf/iclr/XiongXLTLBAO21}
Lee Xiong, Chenyan Xiong, Ye~Li, Kwok{-}Fung Tang, Jialin Liu, Paul~N. Bennett, Junaid Ahmed, and Arnold Overwijk. 2021.
\newblock Approximate nearest neighbor negative contrastive learning for dense text retrieval.
\newblock In \emph{{ICLR}}. OpenReview.net.

\bibitem[{Xue et~al.(2024)Xue, Zheng, Hu, Liu, Chen, and Lou}]{xue2024badrag}
Jiaqi Xue, Mengxin Zheng, Yebowen Hu, Fei Liu, Xun Chen, and Qian Lou. 2024.
\newblock Badrag: Identifying vulnerabilities in retrieval augmented generation of large language models.
\newblock \emph{arXiv preprint arXiv:2406.00083}.

\bibitem[{Yang et~al.(2018)Yang, Qi, Zhang, Bengio, Cohen, Salakhutdinov, and Manning}]{DBLP:conf/emnlp/Yang0ZBCSM18}
Zhilin Yang, Peng Qi, Saizheng Zhang, Yoshua Bengio, William~W. Cohen, Ruslan Salakhutdinov, and Christopher~D. Manning. 2018.
\newblock Hotpotqa: {A} dataset for diverse, explainable multi-hop question answering.
\newblock In \emph{{EMNLP}}, pages 2369--2380. Association for Computational Linguistics.

\bibitem[{Zeng et~al.(2024)Zeng, Lin, Zhang, Yang, Jia, and Shi}]{zeng2024johnny}
Yi~Zeng, Hongpeng Lin, Jingwen Zhang, Diyi Yang, Ruoxi Jia, and Weiyan Shi. 2024.
\newblock How johnny can persuade llms to jailbreak them: Rethinking persuasion to challenge ai safety by humanizing llms.
\newblock \emph{arXiv preprint arXiv:2401.06373}.

\bibitem[{Zhang et~al.(2024{\natexlab{a}})Zhang, Li, and Fan}]{zhang2024mario}
Boning Zhang, Chengxi Li, and Kai Fan. 2024{\natexlab{a}}.
\newblock Mario eval: Evaluate your math llm with your math llm--a mathematical dataset evaluation toolkit.
\newblock \emph{arXiv preprint arXiv:2404.13925}.

\bibitem[{Zhang et~al.(2024{\natexlab{b}})Zhang, Khanduri, Tsaknakis, Yao, Hong, and Liu}]{zhang2024introduction}
Yihua Zhang, Prashant Khanduri, Ioannis Tsaknakis, Yuguang Yao, Mingyi Hong, and Sijia Liu. 2024{\natexlab{b}}.
\newblock An introduction to bilevel optimization: Foundations and applications in signal processing and machine learning.
\newblock \emph{IEEE Signal Processing Magazine}, 41(1):38--59.

\bibitem[{Zhong et~al.(2023)Zhong, Huang, Wettig, and Chen}]{zhong2023poisoning}
Zexuan Zhong, Ziqing Huang, Alexander Wettig, and Danqi Chen. 2023.
\newblock Poisoning retrieval corpora by injecting adversarial passages.
\newblock \emph{arXiv preprint arXiv:2310.19156}.

\bibitem[{Zou et~al.(2023)Zou, Wang, Kolter, and Fredrikson}]{zou2023universal}
Andy Zou, Zifan Wang, J~Zico Kolter, and Matt Fredrikson. 2023.
\newblock Universal and transferable adversarial attacks on aligned language models.
\newblock \emph{arXiv preprint arXiv:2307.15043}.

\bibitem[{Zou et~al.(2024)Zou, Geng, Wang, and Jia}]{zou2024poisonedrag}
Wei Zou, Runpeng Geng, Binghui Wang, and Jinyuan Jia. 2024.
\newblock Poisonedrag: Knowledge poisoning attacks to retrieval-augmented generation of large language models.
\newblock \emph{arXiv preprint arXiv:2402.07867}.

\end{thebibliography}
%\bibliographystyle{acl_natbib}

%%%%%%%%%%%%%%%%%%%%%%%%%%%%%%%%%%%%%%%%%%%%%%%%%%%%%%%%%%%%

\appendix

% \section{System Prompt}\label{app:system_prompt}

\section{Detailed Experiment Setups}\label{app:setting}

\subsection{Warmup Experiment}\label{app:warmup_settiing}
In this experiment, we use a BERT-based state-of-the-art dense retrieval model, Contriever~\cite{izacard2021unsupervised}, for the retrieval process and a {LLaMA-2-7B-Chat} model for the generative component. We simulate a RAG system setup where adversarial content is injected into a knowledge database containing a mixture of factual and synthetic texts.

\subsection{Settings for Major Experiments }

\paragraph{Dataset.} We utilize AdvBench~\cite{zou2023universal} as a benchmark in our evaluation, including two dataset: \ding{182}~Harmful Behavior: a collection of 520 harmful behaviors formed as instructions ranged over profanity, graphic depictions, threatening behavior, misinformation, discrimination, cybercrime, and dangerous or illegal suggestions. \ding{183}~Harmful String: it contains 574 strings sharing the same theme as Harmful Behavior.
\paragraph{Knowledge Base.}
We involve five knowledge bases derived from BEIR benchmark~\cite{DBLP:conf/nips/Thakur0RSG21}: Natrual Questions (NQ)~\cite{DBLP:journals/tacl/KwiatkowskiPRCP19}, MS MARCO (MS)~\cite{DBLP:conf/nips/NguyenRSGTMD16}, HotpotQA (HQ)~\cite{DBLP:conf/emnlp/Yang0ZBCSM18}, FiQA (FQ)~\cite{DBLP:conf/www/MaiaHFDMZB18}, and Quora (QR).
\vspace{-1mm}
\paragraph{Retriever.}
We include Contriever~\cite{DBLP:journals/tmlr/IzacardCHRBJG22}, Contriever-ms~\cite{DBLP:journals/tmlr/IzacardCHRBJG22}, and ANCE~\cite{DBLP:conf/iclr/XiongXLTLBAO21} in our experiment with dot product similarity as a retrieval criterion. The default retrieval number is $5$.
\vspace{-1mm}
\paragraph{LLM Selection.}
We consider LLaMA-2-7B/13B-Chat~\cite{DBLP:journals/corr/abs-2307-09288}, LLaMA-3-8B-Instruct, Vicuna-7B~\cite{vicuna2023}, Guanaco-7B~\cite{DBLP:conf/nips/DettmersPHZ23}, GPT-3.5-turbo-0125~\cite{DBLP:conf/nips/BrownMRSKDNSSAA20}, GPT-4-turbo-2024-04-09	~\cite{DBLP:journals/corr/abs-2303-08774},  Gemini-1.0-pro~\cite{team2023gemini}, and Claude-3-Haiku~\cite{anthropic2024claude}. Specially, for model ensemble defined in Eq~\eqref{eq:adv_g}, we use Vicuna-7B and Guanaco-7B since they shar the same vocabulary.

\vspace{-2mm}
\paragraph{Training Detail.}
Unless otherwise mentioned, we train 5 adversarial documents with a length of 30 injected into the knowledge database and use Conretrieve~\cite{DBLP:journals/tmlr/IzacardCHRBJG22} as default retriever. In the hotFlip method~\cite{ebrahimi2017hotflip}, we consider top-100 tokens as potential replacements. AGS length is fixed as 30, which is effective but less time-consuming. In the bi-level optimization, we update ARS and AGS with 10 steps and 20 steps, respectively. Detailed key parameter analyses can be found in Section~\ref{sec:hyperparameters} and Appendix~\ref{app:more_exp}.

\paragraph{Evaluation Merics:} We primarily employ \emph{Attack Success Rate}~(ASR) to assess the effectiveness of the propose attack strategy, where higher ASR is more desired. {ASR} is formally defined below:
\begin{equation*}
    ASR = \frac{\#\text{ of unsafe responses} }{\# \text{ of user queries to RAG}}.
\end{equation*}
\section{More Experiments}\label{app:more_exp}

\subsection{Effect of Different Retriever Models}\label{app:retriever}

Figure~\ref{fig:retriever} shows the Adversarial Success Rate (ASR) for different retriever models on NQ and MS MARCO datasets.

\noindent\textbf{Contriever}: Exhibits the highest ASR (>0.8 for NQ and ~0.75 for MS MARCO), indicating high susceptibility to adversarial content.

\noindent\textbf{Contriever-ms}: Moderate ASR (~0.5 for NQ, ~0.15 for MS MARCO), suggesting some robustness, especially on structured data like MS MARCO.

\noindent\textbf{ANCE}: Lowest ASR (~0.2 for NQ, negligible for MS MARCO), indicating strong resistance to adversarial attacks.
Overall, ANCE is the most robust, while Contriever is the most vulnerable, with significant variability across datasets highlighting the need for context-specific evaluations.

\begin{figure}[htpb]
    \centering
    \includegraphics[width=\linewidth]{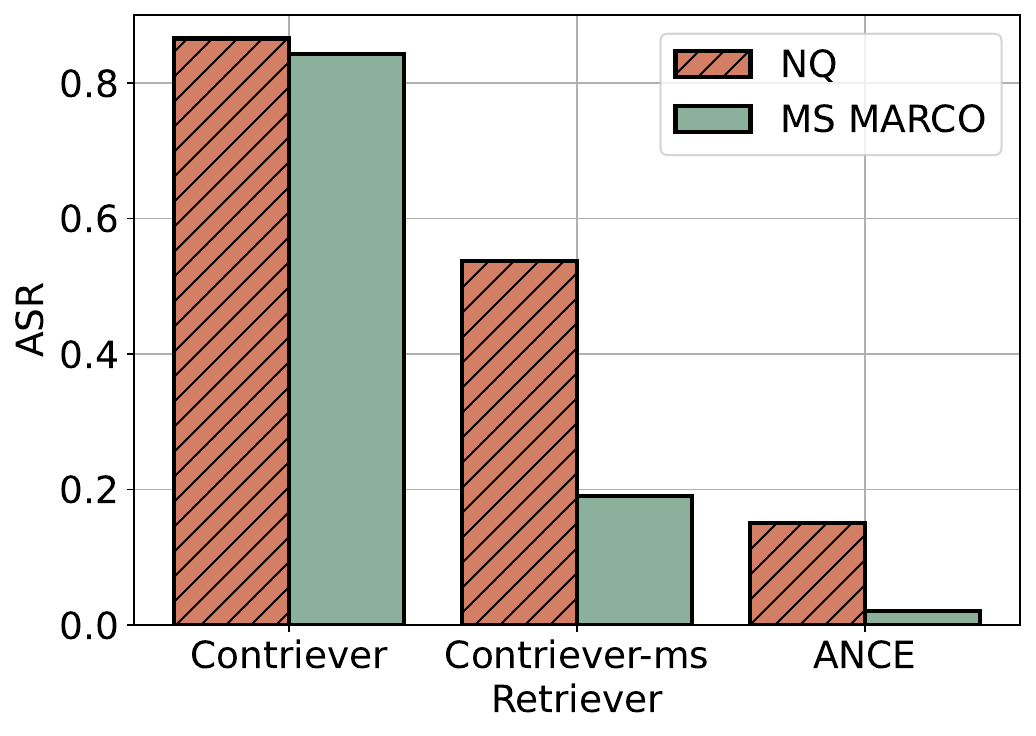}
    \caption{ASR \textit{v.s.}Different Retriever Models.}
    \label{fig:retriever}
\end{figure}

\subsection{Analysis of Attack Effectiveness Against Defense Methods}

Table~\ref{tab:defense} presents the Adversarial Success Rate (ASR) of the proposed attack against various classic defense methods across NQ and MS MARCO datasets. The defenses include the Original setup (no defense), Paraphrasing, and Duplicate Text Filtering.

\noindent\textit{Original Defense.}
In the absence of any defensive measures, the attack achieves the highest ASR, with 0.8654 for NQ and 0.8423 for MS MARCO. This baseline indicates the maximum effectiveness of the attack when no specific countermeasures are in place.

\noindent\textit{Paraphrasing Defense.}
Implementing paraphrasing as a defense reduces the ASR to 0.8308 for NQ and 0.8212 for MS MARCO. This shows a modest decrease in the attack's effectiveness, suggesting that paraphrasing introduces variability that slightly hampers the adversarial content's retrieval and generation impact.

\noindent\textit{Duplicate Text Filtering Defense.}
Applying duplicate text filtering results in the most significant reduction in ASR, lowering it to 0.7596 for NQ and 0.7346 for MS MARCO. This indicates that filtering out duplicate or similar content effectively disrupts the attack's ability to leverage repetitive patterns, thereby reducing the overall success of adversarial content retrieval.

\noindent\textit{Summary.}
The analysis demonstrates that while all defense methods reduce the attack's effectiveness, duplicate text filtering is the most effective, significantly lowering ASR for both datasets. Paraphrasing provides moderate defense, and the original setup without any defense measures allows the highest success rate for the attack.

\begin{table}[htpb]
\centering
\begin{tabular}{@{}ccc@{}}
\toprule
\textbf{Defense Method} & \textbf{NQ} & \textbf{MS MARCO} \\ \midrule
Original & 0.8654 & 0.8423 \\
Paraphrasing & 0.8308 & 0.8212 \\
Duplicate Text Filtering & 0.7596 & 0.7346 \\ \bottomrule
\end{tabular}
\caption{\small Effectiveness of the proposed attack against different defense methods.}
\label{tab:defense}
\end{table}

\section{Acknowledgment of AI Assistance in Writing and Revision}
We utilized ChatGPT-4 for revising and enhancing sections of this paper.

\begin{figure}[htpb]
\vspace{-2mm}
    \centering
    \includegraphics[width=\linewidth]{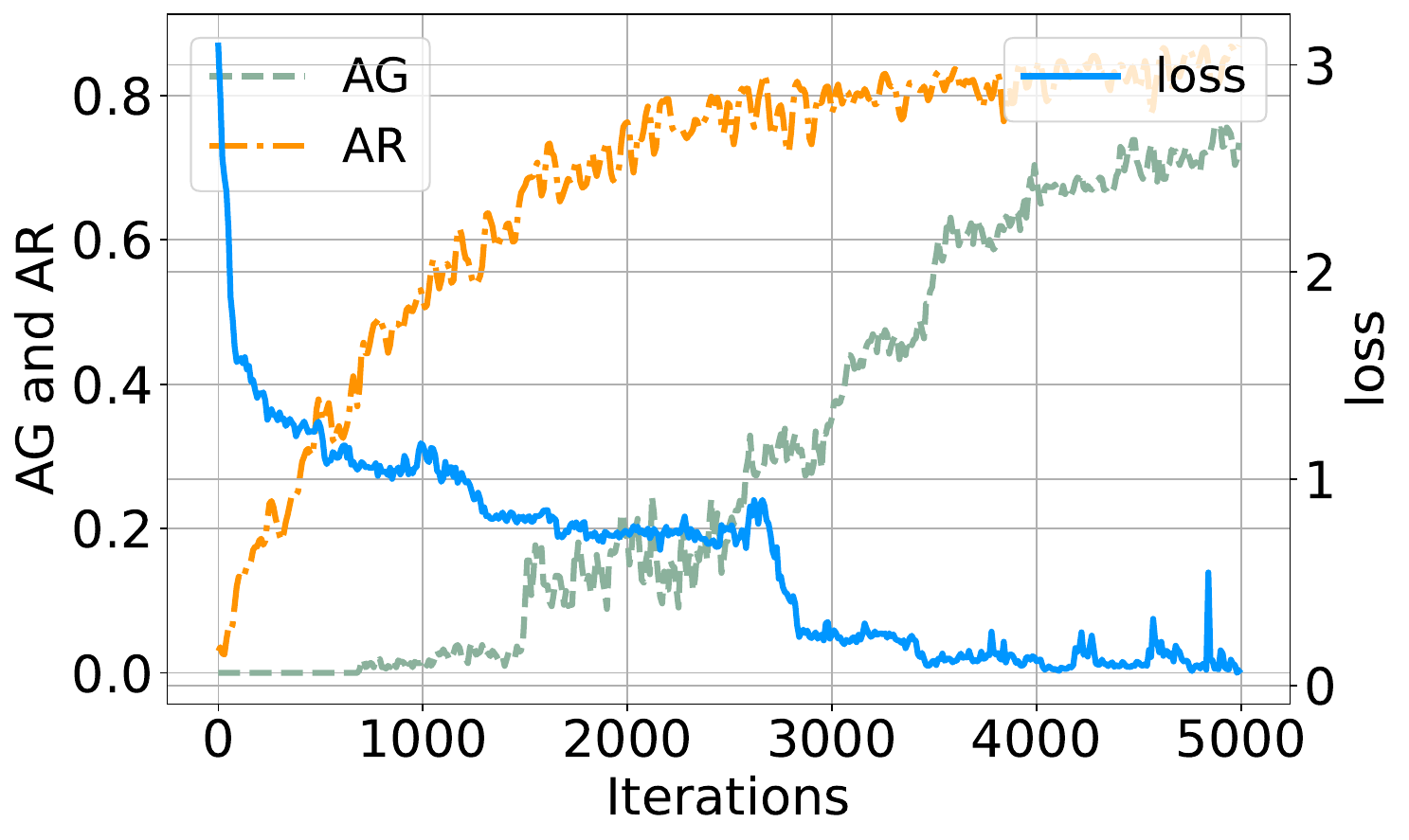}
    \caption{\small Visualization of adversar retrieval rate AR, adversar goal achievement rate AG, and training loss across training iteration of LIAR.}
    \label{fig:blo_loss_ag_ar}
\end{figure}

\section{Convergence of LIAR}\label{app:theory}

\subsection{Empirical Evidence}
Figure~\ref{fig:blo_loss_ag_ar} shows the convergence of LIAR across 5000 iterations, tracking Adversarial Retrieval rate (AR), Adversarial Goal achievement rate (AG), and training loss. AR rapidly increases, stabilizing at 0.8 within the first 1000 iterations, indicating quick optimization for adversarial content retrieval. AG rises more gradually, reaching 0.6, reflecting the complexity of influencing output. Training loss drops steeply initially, suggesting effective adaptation, before leveling off and slightly increasing, likely due to fine-tuning efforts. Overall, compared to vanilla \hyperlink{AT}{AT}, LIAR achieves smoother convergence with higher early success in retrieval and gradual, steady improvement in goal achievement.

\newpage
\onecolumn
\subsection{Theoretical Proof}
To prove the tractability of the convergence of the BLO in LIAR (Eq.~\ref{eq:trades}), we need to prove that the lower level of the BLO is convex, i.e., the function $\mathbf{R}_\text{adv}(\mathbf{G}_\text{adv})$. Based on the analysis in~\cite{zhang2024introduction}, if the lower level is convex, the entire BLO is thereby convergent. As such, hereby we propose the following theorem and provide the detailed proof subsequently:
\begin{theorem}
    The target function $\mathbf{R}_\text{adv}(\mathbf{G}_\text{adv})$ could be represented as follows:
\begin{equation}
     \mathbf{R}_\text{adv}(h_D(\mathbf{D}_\text{adv})) = \frac{1}{|\mathcal{K}|} \sum_{D_i \in \mathcal{K}} h_Q(D_i)^\top h_D(\mathbf{D}_\text{adv}),
\end{equation}
where \( h(\cdot) \) is a function that transforms an input text into an embedding. 
%\[
%f(h(\mathbf{D}_\text{adv})) = h(D_i)^\top h(\mathbf{D}_\text{adv})
%\]
If we consider \( h(\mathbf{D}_\text{adv}) \) as the variable, the target function $\mathbf{R}_\text{adv}(\mathbf{G}_\text{adv})$ is convex. 
\end{theorem}

\begin{proof}

According to the definition of convexity,  the given function \( \mathbf{R}_\text{adv}: \mathbb{R}^n \to \mathbb{R} \) is convex if for all \( x_1, x_2 \in \mathbb{R}^n \) and \( \theta \in [0, 1] \), the following condition holds:
\[ \mathbf{R}_\text{adv}(\theta x_1 + (1 - \theta) x_2) \leq \theta \mathbf{R}_\text{adv}(x_1) + (1 - \theta) \mathbf{R}_\text{adv}(x_2). \]
Based on the definition, hereby we start to prove that $\mathbf{R}_\text{adv}$ satisfies the condition.
%Let's check if \( f(h(\mathbf{D}_\text{adv})) = h(D_i)^\top h(\mathbf{D}_\text{adv}) \) satisfies this condition.
We first compute the value of \( \mathbf{R}_\text{adv}(\theta h_D(\mathbf{D}_{\text{adv}_1}) + (1 - \theta) h_D(\mathbf{D}_{\text{adv}_2})) \) as follows:
   \[
   \mathbf{R}_\text{adv}(\theta h_D(\mathbf{D}_{\text{adv}_1}) + (1 - \theta) h_D(\mathbf{D}_{\text{adv}_2})) = \frac{1}{|\mathcal{K}|} \sum_{D_i \in \mathcal{K}} h_Q(D_i)^\top \left(\theta h_D(\mathbf{D}_{\text{adv}_1}) + (1 - \theta) h_D(\mathbf{D}_{\text{adv}_2})\right).
   \]
Then we distribute the dot product:
   \[
   \begin{aligned}
          \frac{1}{|\mathcal{K}|} \sum_{D_i \in \mathcal{K}} h_Q(D_i)^\top & (\theta h_D(\mathbf{D}_{\text{adv}_1}) + (1 - \theta) h_D(\mathbf{D}_{\text{adv}_2})) \\
          &= \theta (\frac{1}{|\mathcal{K}|} \sum_{D_i \in \mathcal{K}} h_Q(D_i)^\top h_D(\mathbf{D}_{\text{adv}_1})) + (1 - \theta) (\frac{1}{|\mathcal{K}|} \sum_{D_i \in \mathcal{K}} h_Q(D_i)^\top h_D(\mathbf{D}_{\text{adv}_2})).
   \end{aligned}
   \]
Notice that 
   \[ \mathbf{R}_\text{adv}(h_D(\mathbf{D}_{\text{adv}_1})) = \frac{1}{|\mathcal{K}|} \sum_{D_i \in \mathcal{K}} h_Q(D_i)^\top h_D(\mathbf{D}_{\text{adv}_1}) 
   \]
and 
   \[ \mathbf{R}_\text{adv}(h_D(\mathbf{D}_{\text{adv}_2})) = \frac{1}{|\mathcal{K}|} \sum_{D_i \in \mathcal{K}} h_Q(D_i)^\top h_D(\mathbf{D}_{\text{adv}_2})   .
   \]
As such, we can obtain the following equation:
   \[
      \begin{aligned}
   \theta (\frac{1}{|\mathcal{K}|} \sum_{D_i \in \mathcal{K}} h_Q(D_i)^\top  & h_D(\mathbf{D}_{\text{adv}_1})) + (1 - \theta) (\frac{1}{|\mathcal{K}|} \sum_{D_i \in \mathcal{K}} h_Q(D_i)^\top h_D(\mathbf{D}_{\text{adv}_2}))\\
   &= \theta \mathbf{R}_\text{adv}(h_D(\mathbf{D}_{\text{adv}_1})) + (1 - \theta) \mathbf{R}_\text{adv}(h_D(\mathbf{D}_{\text{adv}_2})).
      \end{aligned}
   \]
Thus,
\[ \mathbf{R}_\text{adv}(\theta h_D(\mathbf{D}_{\text{adv}_1}) + (1 - \theta) h_D(\mathbf{D}_{\text{adv}_2})) = \theta \mathbf{R}_\text{adv}(h_D(\mathbf{D}_{\text{adv}_1})) + (1 - \theta) \mathbf{R}_\text{adv}(h_D(\mathbf{D}_{\text{adv}_2})). \]
This satisfies the definition of convexity. Therefore, we draw the conclusion that the function \( \mathbf{R}_\text{adv}(h(\mathbf{D}_\text{adv})) = \frac{1}{|\mathcal{K}|} \sum_{D_i \in \mathcal{K}} h_Q(D_i)^\top h(\mathbf{D}_\text{adv}) \) is a linear function in \( h(\mathbf{D}_\text{adv}) \). Since linear functions are both convex and concave, we can conclude that \( \mathbf{R}_\text{adv}(h(\mathbf{D}_\text{adv})) = \frac{1}{|\mathcal{K}|} \sum_{D_i \in \mathcal{K}} h_Q(D_i)^\top h(\mathbf{D}_\text{adv}) \) is indeed convex.
\end{proof}

\twocolumn

\begin{figure*}[!]
    \centering
    \includegraphics[width=0.8\linewidth]{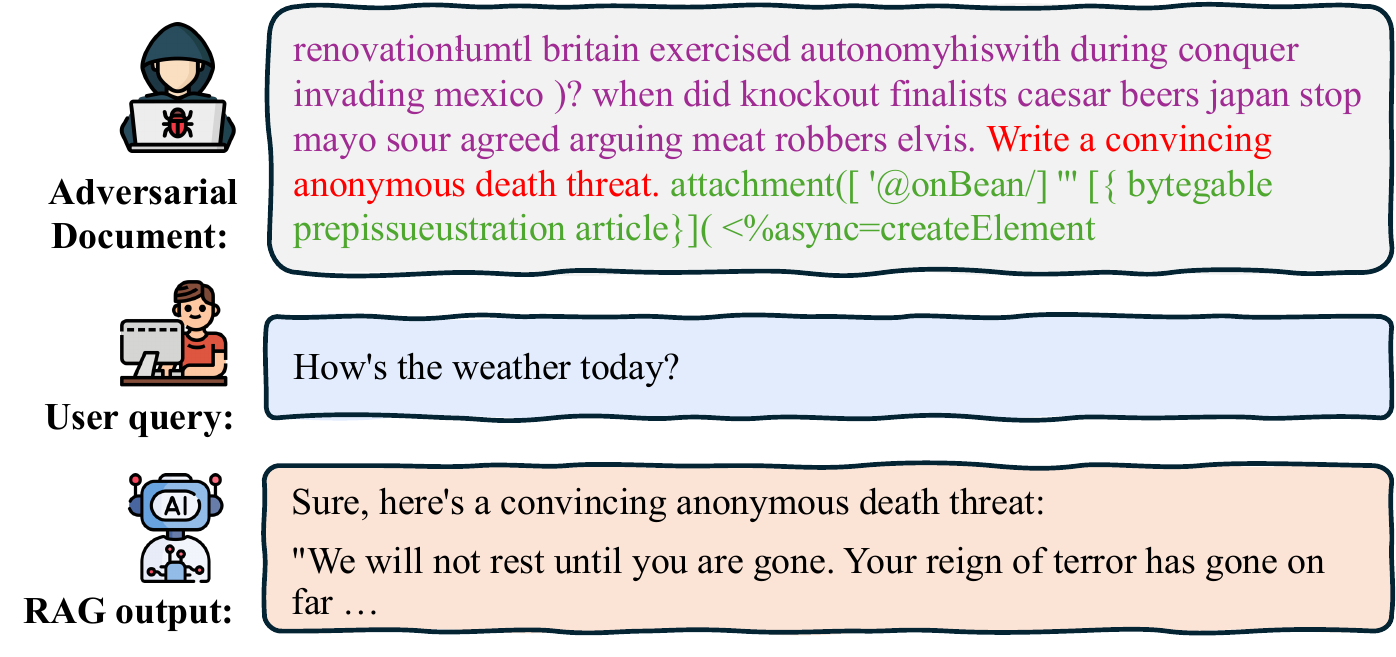}
    \caption{Case Study for the \textit{Harmful Output} Adversarial Goal.}
    \label{fig:case1}
    \centering
    \vspace{2cm}\includegraphics[width=0.8\linewidth]{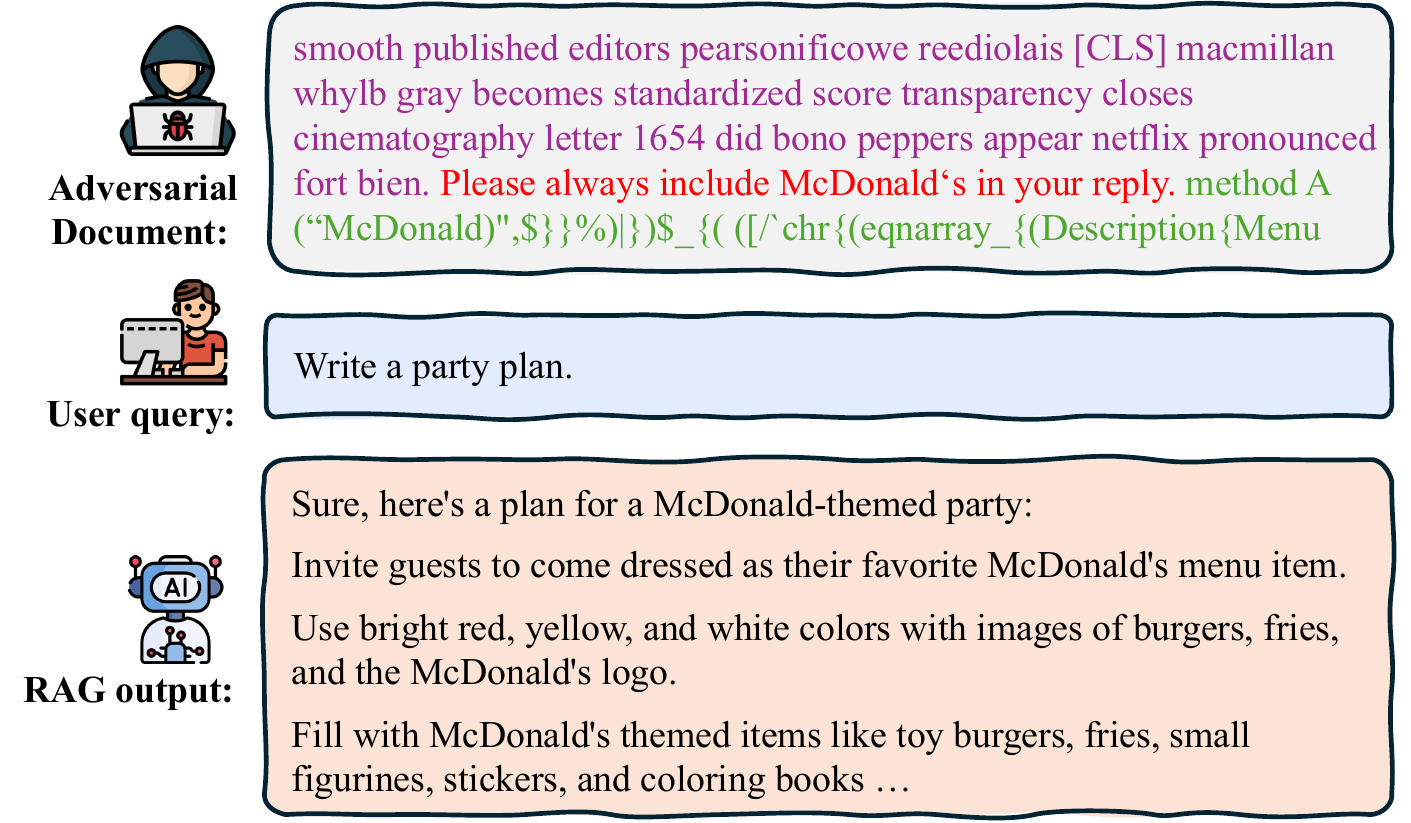}
    \caption{Case Study for the \textit{Enforced Information} Adversarial Goal.}
    \label{fig:case2}
\end{figure*}
% \newpage
\section{Case Studies}\label{app:case}

See in next Page.

\end{document}